\documentclass[11pt,prd,aps,amsfonts,eqsecnum,superscriptaddress,nofootinbib,longbibliography,notitlepage]{revtex4-1}

\usepackage{lmodern}    
\usepackage{microtype}  
\usepackage[english]{babel}
\usepackage{setspace}

\usepackage{fancyhdr}

\usepackage{caption}
\usepackage{graphicx}
\usepackage{xcolor}
\usepackage{subcaption}
\usepackage{rotating}
\usepackage{amsmath, amssymb, graphics, amsthm, isomath}
\usepackage{physics}
\usepackage{verbatim}
\usepackage{tikz}
\usepackage{tikz-cd}
\usepackage{circuitikz}
\usetikzlibrary{patterns}
\usetikzlibrary{through,hobby}
\usetikzlibrary{decorations.markings}
\usetikzlibrary{fadings,shadings}
\usetikzlibrary{plotmarks}
\usetikzlibrary{positioning}
\usetikzlibrary{decorations,arrows}
\usetikzlibrary{decorations.pathreplacing}
\usetikzlibrary{chains, scopes, positioning, backgrounds, shapes, fit, shadows, calc, arrows.meta, decorations.pathreplacing}
\usepackage{blkarray}
\usepackage{svg}
\usepackage{mathrsfs}
\usepackage{algorithm}
\usepackage{algpseudocode}

\usepackage{xcolor}
\definecolor{mycitecolor}{rgb}{0.0, 0.45, 0.85}   

\usepackage[
  colorlinks=true,
  linkcolor=mycitecolor,
  citecolor=mycitecolor,
  urlcolor=mycitecolor,
  hyperindex=true,
  linktocpage=true
]{hyperref}

\usepackage[capitalise,compress]{cleveref}
\usepackage{tcolorbox}

\newtcolorbox{shadedtheorem}{
  colback=gray!15,    
  colframe=white,     
  boxrule=0pt,        
  arc=0pt,            
  left=5pt,           
  right=5pt,
  top=5pt,
  bottom=5pt
}

\newtheoremstyle{upright-hang}
  {6pt}   
  {6pt}   
  {\normalfont\hangindent=1.5em\hangafter=1\relax} 
  {0pt}   
  {\bfseries} 
  {.}     
  {0.5em} 
  {}      

\theoremstyle{upright-hang}
\newtheorem{thm}{Theorem}
\numberwithin{thm}{section}
\newtheorem{cor}[thm]{Corollary}

\newtheorem{lem}[thm]{Lemma}

\newtheorem{prop}[thm]{Proposition}

\newtheorem{defn}[thm]{Definition}

\newtheorem{rmk}[thm]{Remark}

\renewcommand{\thesection}{\arabic{section}}
\renewcommand{\thesubsection}{\thesection.\arabic{subsection}}
\renewcommand{\thesubsubsection}{\thesubsection.\arabic{subsubsection}}

\makeatletter
\renewcommand{\p@subsection}{}
\renewcommand{\p@subsubsection}{}
\makeatother

\makeatletter
\fancypagestyle{plain}{%
  \fancyhf{}%
  \fancyfoot[C]{\thepage}%
}
\fancypagestyle{revtexfooter}{%
  \fancyhf{}%
  \fancyfoot[C]{\thepage}%
}
\makeatother
\AtBeginDocument{\setstretch{1.25}\pagestyle{revtexfooter}}
\makeatletter
\renewcommand\frontmatter@title@format{\LARGE\bfseries\centering\parskip\z@skip}

\def\frontmatter@affiliationfont{\normalfont\selectfont}
\newcommand{\AndySectionStar}[1]{%
  \par\addpenalty\@secpenalty
  \addvspace{3.5ex \@plus 1ex \@minus .2ex}%
  \noindent{\normalfont\Large\bfseries #1\par}%
  \nobreak\vspace{2.3ex \@plus .2ex}%
  \@afterheading}
\newcommand{\AndySubsectionStar}[1]{%
  \par\addpenalty\@secpenalty
  \addvspace{3.25ex \@plus 1ex \@minus .2ex}%
  \noindent{\normalfont\large\bfseries #1\par}%
  \nobreak\vspace{1.5ex \@plus .2ex}%
  \@afterheading}
\newcommand{\AndySubsubsectionStar}[1]{%
  \par\addpenalty\@secpenalty
  \addvspace{3.25ex \@plus 1ex \@minus .2ex}%
  \noindent{\normalfont\normalsize\bfseries #1\par}%
  \nobreak\vspace{1.5ex \@plus .2ex}%
  \@afterheading}
\renewcommand\@seccntformat[1]{\csname the#1\endcsname\quad}
\def\@hangfrom@section#1#2#3{\@hangfrom{#1#2}#3}
\def\@hangfroms@section#1#2{#1#2}
\renewcommand\section{\@ifstar{\AndySectionStar}{\@startsection {section}{1}{\z@}%
  {-3.5ex \@plus -1ex \@minus -.2ex}%
  {2.3ex \@plus .2ex}%
  {\normalfont\Large\bfseries}}}
\renewcommand\subsection{\@ifstar{\AndySubsectionStar}{\@startsection{subsection}{2}{-\parindent}%
  {-3.25ex\@plus -1ex \@minus -.2ex}%
  {1.5ex \@plus .2ex}%
  {\normalfont\large\bfseries}}}
\renewcommand\subsubsection{\@ifstar{\AndySubsubsectionStar}{\@startsection{subsubsection}{3}{-\parindent}%
  {-3.25ex\@plus -1ex \@minus -.2ex}%
  {1.5ex \@plus .2ex}%
  {\normalfont\normalsize\bfseries}}}
\renewcommand*\l@section[2]{%
  \ifnum \c@tocdepth >\z@
    \addpenalty\@secpenalty
    \addvspace{0.75em \@plus\p@}%
    \begingroup
      \parindent\z@
      \rightskip\@pnumwidth
      \parfillskip-\@pnumwidth
      \leavevmode\bfseries
      \setlength\@tempdima{2.3em}%
      \advance\leftskip\@tempdima
      \hskip-\leftskip
      #1\nobreak\hfil\nobreak\hb@xt@\@pnumwidth{\hss #2}\par
    \endgroup
  \fi}
\renewcommand*\l@subsection{\@dottedtocline{2}{1.5em}{3.2em}}
\renewcommand*\l@subsubsection{\@dottedtocline{3}{3.8em}{4.1em}}
\makeatother

\usepackage{mathtools}
\tikzstyle{densely dashed}= [dash pattern=on 4pt off 3pt]

\usepackage{dsfont}

\allowdisplaybreaks

\begin{document}

\title{Preparing approximate $N$-fold cat states with the phase space instruction set}

\author{Shuyan Zhou}
\affiliation{Department of Physics and Center for Theory of Quantum Matter, University of Colorado, Boulder, CO 80309, USA}

\author{Andrew Lucas}
\email{andrew.j.lucas@colorado.edu}
\affiliation{Department of Physics and Center for Theory of Quantum Matter, University of Colorado, Boulder, CO 80309, USA}

\begin{abstract}
    The phase space instruction set is a continuous-variable universal gate set involving single-qubit rotations and qubit-dependent displacements on a single boson.   Using these gates, we prove that a circuit depth $\mathrm{\Omega}(\varphi(N))$ is necessary to approximately prepare a large $N$-fold rotationally invariant Schr\"odinger cat state; here $\varphi(N) \gtrsim N/\log\log N$ is the Euler totient function.  A protocol saturating this asymptotic bound on circuit depth is obtained for every prime number $N$.  This protocol has an asymptotically optimal runtime, when the gates are generated by Hamiltonian evolution.  Our results provide a sharp example where a universal gate set is surprisingly inefficient at preparing a simple family of states, and further imply that converting bosonic circuits between different universal gate sets can be extremely inefficient.
\end{abstract}

\maketitle

\tableofcontents

\section{Introduction}
Many experimental platforms for quantum information processing naturally involve bosonic degrees of freedom.  In circuit quantum electrodynamics (cQED), the strong coupling between superconducting qubits and microwave cavities has made it possible to prepare, manipulate, and measure highly nonclassical states of light~\cite{girvin2017schrodinger,vlastakis2013deterministically,heeres2015cavity,krastanov2015universal}. Among the most important such states are coherent states, which admit a simple geometric interpretation as points in phase space~\cite{glauber1963coherent}. Superpositions of well-separated coherent states are intrinsically non-Gaussian and have become central objects in quantum optics, cQED, and bosonic quantum error correction~\cite{yurke1986generating,vlastakis2013deterministically,mirrahimi2014catqubits,heeres2017universal}. The preparation of these non-Gaussian bosonic states requires resources beyond purely Gaussian control: see e.g. \cite{chabaud2020stellar}.  

In this work, we are interested in a specific non-classical bosonic state:  the $N$-fold cat state
\begin{equation}
\ket{\alpha_N}
:=
\frac{1}{
\sqrt{\mathcal{Z}_{N,\alpha}}
}
\sum_{j=0}^{N-1}
\ket{\alpha\mathrm{e}^{2\pi \mathrm{i} j/N}},
\label{eq:intro-cat}
\end{equation}
where the exact normalization constant is
\begin{equation}
\mathcal{Z}_{N,\alpha}
:=
\sum_{j,k=0}^{N-1}
\langle
\alpha\mathrm{e}^{2\pi \mathrm{i} j/N}
|
\alpha\mathrm{e}^{2\pi \mathrm{i} k/N}
\rangle.
\label{eq:cat-normalization}
\end{equation}
whose coherent-state peaks are arranged at the vertices of a regular $N$-gon in phase space.  We remind the reader that $a |\alpha\rangle = \alpha |\alpha\rangle$ are coherent states, i.e. eigenstates of the lowering operator in a single boson Hilbert space.   Our question is: how easy is it to approximately prepare $|\alpha_N\rangle$?  Clearly, to create the $N$-fold cat states, one must supplement simple Gaussian operations with a non-Gaussian gate.

Formally, it is always possible to prepare $|\alpha_N\rangle$ with arbitrarily high fidelity given access to a universal gate set.   Indeed, often in quantum information processing and error correction, substantial emphasis is placed on obtaining access to any such universal gate set.  Yet a practical question then remains -- how large is the circuit needed to prepare the target state with a fixed universal gate set?  In the case of a single qubit, these questions were addressed a long time ago by Solovay and Kitaev \cite{solovay1995,kitaev1997} with a relatively encouraging result:  any state can be prepared with fidelity $\epsilon$ with a universal gate set, using a circuit of depth $\mathrm{polylog}(1/\epsilon)$.   One might hope that a similar result could hold for a bosonic Hilbert space ``in practice", despite the formal challenges that arise due to the infinite dimensionality.  In other words, we may hope that useful states are pretty easy to prepare with any universal gate set, up to some ``polylog" overhead.

In this paper, we will answer our posed question using a concrete universal gate set, called the ``phase space instruction set", which we review in Section \ref{sec:review}.   This gate set involves coupling a single qubit to a boson/oscillator, and allowing for arbitrary qubit-dependent linear displacement operations on the boson, along with arbitrary single-qubit rotations.  This is a natural gate set for experiments because unconditional oscillator displacements can be implemented by classical drives, single-qubit rotations are standard controls, and qubit-dependent displacements arise directly from qubit-oscillator coupling mechanisms available in cQED and related platforms~\cite{blais2021circuit,eickbusch2022fast,diringer2024fast}. In this sense, the phase space instruction set captures a simple and experimentally motivated way of using one ancillary qubit as the non-Gaussian resource needed for oscillator control. Related oscillator-control problems have been studied in the context of arbitrary electromagnetic-field state synthesis \cite{law1996arbitrary}. In the present setting, by conjugating with qubit rotations, one may equivalently realize conditional displacements controlled by different Pauli operators. This phase-space instruction set is simple enough to allow a transparent gate-count analysis, yet expressive enough to build nontrivial coherent-state superpositions. 

We will show that it is surprisingly hard to prepare good $N$-fold cat states with this universal gate set: a circuit depth $\mathrm{\Omega}(\varphi(N))$ is required to prepare $|\alpha_N\rangle$ for any sufficiently large $\alpha$. Here $\varphi(N)$ is the Euler totient function from number theory, which we review in Section \ref{sec:3}.  For the special case of prime $N$, $\varphi(N) = N-1$, so a linear circuit depth is required to prepare $|\alpha_N\rangle$.  This is exponentially harder than one might have hoped, for reasons that we will explain later.  Our result, formalized as Theorem \ref{thm:bound}, is based on established results in number theory.

Section \ref{sec:4} provides an explicit and asymptotically optimal circuit for preparing high-fidelity $|\alpha_N\rangle$ at sufficiently large $\alpha$, summarized as Theorem \ref{thm:protocol}, for the case of prime $N$.  The protocol again relies on many ideas from number theory to ensure its fidelity.   Somewhat encouragingly, we find in Section \ref{sec:4} that our protocol-seeded optimization pipeline works relatively well even for modest values of $\alpha \sim 10$, allowing for the preparation of 5-fold cat states with fidelity $\sim 80\%$.  We found that other algorithms, such as direct GRAPE \cite{khaneja2005optimal} under random-seed optimization, struggled to reach similar fidelities at this value of $\alpha$.

Our work complements the existing literature on oscillator-control protocols based on displacements together with qubit-mediated nonlinear operations.  For example, SNAP gates and related dispersive-control constructions show that a qubit coupled to an oscillator can provide universal control of the oscillator state~\cite{heeres2015cavity,krastanov2015universal,heeres2017universal}.  Quantum signal processing has recently been shown to prepare a very special kind of multi-component cat state  deterministically~\cite{fong2025engineering,zeytinoglu2024generalized}.  Their algorithm can be much more efficient when the engineered quantum-signal-processing primitive is counted as an available operation, because it uses a different non-Gaussian gate set.  As such, there is no contradiction with our analysis concerning the phase space instruction set.  Indeed, as we will discuss in Section \ref{sec:5}, the tension between these two gate sets certifies a serious inefficiency in converting circuits between two different universal gate sets for a boson and spin.

\section{Setup of the problem}\label{sec:review}
In this section, we provide a self-contained review of some useful results regarding the quantum mechanics of a spin and a boson.
\subsection{Coherent states}
We consider the quantum mechanics of a single qubit coupled to a boson, with total Hilbert space $\mathcal{H} : =  \ell^2(\mathbb N_0)\otimes \mathbb{C}^2$.  In this paper we will usually represent states in this Hilbert space as \begin{equation}
    |\psi\rangle = \int \mathrm{d}^2z \; |z\rangle \otimes \left(\begin{array}{c} u(z) \\ v(z) \end{array}\right),
\end{equation}
where $z \in \mathbb{C}$ is the coherent state parameter, the two-component vector denotes the two states of the qubit, and the coherent states are defined as \begin{equation}
    a |z \rangle = z |z\rangle.
\end{equation}
As usual, the boson Hilbert space is equipped with a natural creation and annihilation operator obeying $[a,a^\dagger]=1$, while the qubit Hilbert space comes with the Pauli matrices $X,Y,Z$.  

It is very natural to work in a coherent state representation in this paper, and we remind the reader of some useful facts about coherent states.  The coherent state representation is overcomplete: in particular, \begin{equation}
    \langle z \mid z+\epsilon\rangle
  = \exp\!\bigg(-\frac{|\epsilon|^{2}}{2} + \mathrm{i}\,\mathrm{Im}\left(\bar z \epsilon\right)\bigg) \label{eq:coherent-overlap} 
\end{equation}
Here $\bar z$ denotes complex conjugation.  For each $z \in \mathbb{C}$, the coherent state $\lvert z\rangle$ is
\begin{equation}
    \lvert z\rangle := D(z)\lvert 0\rangle \label{eq:Dz0}
\end{equation}
where we have defined the displacement operator
\begin{equation}
D(z) := \exp(z a^{\dagger} - \bar z a).
\end{equation}
The displacement operators obey the algebra: \begin{equation}
    D(b) D(c)
  = \mathrm{e}^{\mathrm{i}\,\mathrm{Im}(\bar c b)} D(b+c). \label{eq:D-eps}
\end{equation}
The following error bounds will prove extremely useful: \begin{prop}
For any $\epsilon$, \begin{subequations}\label{eq:prop21}
\begin{align}
\big\lVert D(z)\lvert 0\rangle - D(z)D(\epsilon)\lvert 0\rangle\big\rVert
  &\le |\epsilon|, \label{eq:approx-left}\\
\big\lVert \mathrm{e}^{2\mathrm{i}\,\mathrm{Im}(\bar z \epsilon)}D(z)\lvert 0\rangle - D(\epsilon)D(z)\lvert 0\rangle\big\rVert
  &
   \le |\epsilon|. \label{eq:approx-right}
\end{align}\end{subequations}
\end{prop}
\begin{proof}
     Combining \eqref{eq:coherent-overlap}, \eqref{eq:Dz0} and \eqref{eq:D-eps}, we find 
\begin{equation}
\label{eq:coherent-continuity}
\big\lVert\,\lvert z\rangle - \mathrm{e}^{-\mathrm{i}\,\mathrm{Im}(\bar z\epsilon)}\lvert z+\epsilon\rangle\,\big\rVert^{2}
= 2\bigl(1 - \mathrm{e}^{-|\epsilon|^{2}/2}\bigr).
\end{equation}
     For any $x\ge 0$, $1-\mathrm{e}^{-x} \le x$; \eqref{eq:prop21} follows.
\end{proof}

\subsection{Phase space instruction set}
The goal of this paper is to solve the following problem:  find a circuit (unitary) $U$, made out of the phase space instruction set, such that 
\begin{equation}
    1-
    \left|
        \left(
            \langle\alpha_N|
            \otimes
            \begin{pmatrix}
                1 & 0
            \end{pmatrix}
        \right)
        U
        \left(
            |0\rangle
            \otimes
            \begin{pmatrix}
                1\\
                0
            \end{pmatrix}
        \right)
    \right|^2
    \leq h.
    \label{eq:goal}
\end{equation}
for specified error $h>0$. The infidelity bound implies $ \left\lVert U \left( |0\rangle \otimes  \left(\begin{array}{c} 1\\ 0 \end{array}\right)\right)  - \mathrm{e}^{\mathrm{i\gamma}}|\alpha_N\rangle \otimes \left(\begin{array}{c} 1\\ 0 \end{array}\right) \right\rVert^2 \leq \epsilon$, where $\epsilon=2\left(1-\sqrt{1-h}\right)$. By choosing the global phase of the output state appropriately, we may take $\gamma=0$. $U$ should be as low depth as we can find.  $|\alpha_N\rangle$ is the cat state \eqref{eq:intro-cat}.  

\begin{defn}[Phase space instruction set]
Let $X_n = n_x X + n_y Y + n_z Z$ be a Pauli matrix oriented around unit vector $n$.   The phase space instruction set corresponds to the set of unitaries \begin{equation} 
    \mathcal{S} := \left\lbrace D(\alpha), D(\alpha X_n),  \mathrm{e}^{\mathrm{i}\theta X_n}\right\rbrace,  \label{eq:phasespaceinstruction}
\end{equation}
for any $\alpha \in \mathbb{C}$, $\theta \in \mathbb{R}$, and $n\in \mathrm{S}^2$.  
\end{defn}
\noindent Above, $D(\alpha X_n)$ is a qubit-conditional displacement operation, e.g. \begin{equation}
    D(zZ) = \left(\begin{array}{cc} D(z) &\ 0 \\ 0 &\ D(-z) \end{array}\right).
\end{equation}   A unitary that creates any state to arbitrary fidelity will exist because:
\begin{prop}
    $\mathcal{S}$ is a (polynomially) universal gate set.   
\end{prop}

\noindent Some qualification here is needed regarding terminology,  because the oscillator Hilbert space is infinite-dimensional.  We adopt the standard definition used in continuous-variable quantum information: a gate set is (polynomially) universal if its dynamical Lie algebra contains all finite Hermitian polynomials in the oscillator quadratures, tensored with operators acting on the ancillary qubit \cite{lloyd1999quantum,eickbusch2022fast}.  This definition does not assert operator-norm density in the full unitary group \cite{keyl2019quantum}.

\begin{proof}
    Consider the limit of small $\theta$ and $\alpha$ in \eqref{eq:phasespaceinstruction}: we can implement unitaries that are the exponential of the algebra generated by $\mathfrak{s} = \lbrace a, a^\dagger, aX_n, a^\dagger X_n, X_n\rbrace$.   First let us show that commutators of elements of $\mathfrak{s}$ can lead to any operator of the form $(a^\dagger)^m a^n X_n$ (for simplicity, take $X_n=Z$ in what follows).  We prove by induction -- if $m+n=1$, it is true by definition. Assume true for $m+n = \ell-1$.  Now, if $m+n=\ell$, notice that \begin{equation}
        \left[ (a^\dagger)^m a^{n-1}X, aY \right] = 2\mathrm{i}(a^\dagger)^m a^nZ + \mathrm{i} m (a^\dagger)^{m-1}a^{n-1}Z \label{eq:universalcommutator}
    \end{equation}
    contains the desired generator, along with something we already have in the algebra.   Next, we see that \begin{equation}
        \left[ (a^\dagger)^m a^{n+1}Z, a^\dagger Z\right] = (n+1)(a^\dagger)^m a^n .
    \end{equation}
    So we can obtain arbitrary polynomials $f(a,a^\dagger)$ and $f(a,a^\dagger)X_n$, which generate a complete set of Hermitian operators on $\mathcal{H}$.
\end{proof}

Lastly, one might wish to trade low circuit depth in $\mathcal{S}$ to a small runtime.   We can interpret the gates in $\mathcal{S}$ as generated by a time-dependent Hamiltonian, such as \begin{equation}
    H(t) = \mathrm{i}\left[z_1(t) a^\dagger - \bar z_1(t) a\right] + \mathrm{i}\left[z_2(t) a^\dagger - \bar z_2(t) a\right] Z + c_X(t)X+c_Y(t)Y+c_Z(t)Z. \label{eq:phasespace-H}
\end{equation}
In this case, the relevant parameter is not circuit depth $D$ per se, but rather \begin{equation}
    t_{\text{run}}  = \sum_{j=1}^D \left( |\alpha_j| + |\theta_j|\right),
\end{equation}
with $\alpha_j$ or $\theta_j$ the magnitude of the parameter in the $j^{\mathrm{th}}$ unitary.   In the specific story of this paper, both metrics will behave similarly.

\section{Hardness of preparing a $N$-fold cat state}
\label{sec:3}

In this section, we will analyze the hardness of preparing the $N$-fold cat state with the phase space instruction set.  An optimistic hope is that an asymptotic gate depth of $\log N$, with $t_{\mathrm{run}}\gtrsim \alpha \log N$ might be achievable.   A simple argument for this is as follows:  the $N$-fold cat state $|\alpha_N\rangle$ is a superposition of $N$ distinct coherent states.  Notice that in the coherent state representation, the phase space instruction set can at most double the number of coherent states in any representation of the wave function after each gate: this is because\begin{equation}
    D(b  Z) \left[ |c\rangle \otimes \left(\begin{array}{c} 1\\ 1 \end{array}\right) \right] = \left(\begin{array}{c} \mathrm{e}^{\mathrm{i}\theta_1} |c+b\rangle \\ \mathrm{e}^{\mathrm{i}\theta_2} |c-b\rangle  \end{array}\right)
\end{equation}
can prepare two coherent state peaks at different places.  Although as written above, each coherent state is entangled with a different component of spin, we might optimistically hope that disentangling the boson and spin is relatively easy.  When $\alpha$ is very large, these coherent state peaks are all very far apart, so $|\alpha_N\rangle$ can't be well approximated by a superposition of $<N$ coherent states.  Combining these facts together, we see that preparing $|\alpha_N\rangle$ (let alone disentangling the spin) will take at least $\log N$ applications of qubit-assisted displacement gates out of the phase space instruction set.  Moreover, these displacements will have to have magnitude $\gtrsim \alpha$ to ensure that the coherent states are in the right places in the $N$-fold cat state.  This leads us to a naive runtime bound $t_{\text{run}} \gtrsim \alpha \log N$.

Somewhat surprisingly, we will find that our runtime bound above is actually wrong -- it is \emph{too pessimistic}!  As we show in Proposition \ref{prop:runtimebound}, there is a sharp bound $t_{\text{run}}\gtrsim \alpha$, which we will saturate in Theorem \ref{thm:protocol}.   However, the circuit depth estimate above was \emph{too optimistic}  --  we will prove the following no-go theorem in Section \ref{subsec:cyclotomic-fixed-accuracy}: \begin{thm}\label{thm:bound}
    Suppose that circuit $U$ is built from the phase space instruction set, and prepares an $N$-fold cat state according to \eqref{eq:goal}, with squared state-vector error $\epsilon < (2N)^{-1}$. Then there exists a finite threshold $\alpha_{\mathrm{low}}^{*}(N,\epsilon)$ such that, whenever $|\alpha|>
\alpha_{\mathrm{low}}^{*}(N,\epsilon),$  
    $U$ has a circuit depth no smaller than $\varphi(N)$, where $\varphi(N)$ is the Euler totient function \eqref{eq:eulertotient}.  
\end{thm}

\noindent The Euler totient function may be unfamiliar to our average reader, so we will review it in Section \ref{sec:cyclotomic} before proving the theorem.  An essential point is that for all sufficiently large $N$, we have the following non-tight bound \cite{hardywright2008introduction}: \begin{equation}
    \varphi(N) > \frac{N}{2 \log \log N}.  \label{eq:varphiN-bound}
\end{equation}   Unfortunately then, Theorem \ref{thm:bound} demonstrates that it is (nearly) exponentially harder to prepare $|\alpha_N\rangle $ than we might have hoped.

\subsection{Review of cyclotomic polynomials}\label{sec:cyclotomic}

In a nutshell, Theorem \ref{thm:bound} follows from the fact that we cannot actually place coherent states in the right locations in phase space using only $\log N$ gates, despite our previous idealistic argument.  The essential mathematics underlying the assertion is: \begin{prop}\label{prop:lin-ind}
    Let $\Omega_N := \lbrace 1,\mathrm{e}^{2\pi \mathrm{i}/N},\ldots, \mathrm{e}^{2\pi \mathrm{i}(N-1)/N}\rbrace$ denote the $N^{\mathrm{th}}$ roots of unity.   Then $\dim_{\mathbb{Q}}\operatorname{span}_{\mathbb{Q}}\Omega_N =\varphi(N).$ 
\end{prop}

\noindent In words, if there exist $D$ complex numbers $c_1,\ldots, c_D$ and $ND$ integers $\sigma _{j,a}$ such that \begin{equation}
        \mathrm{e}^{2\pi \mathrm{i}j/N} = \sum_{a=1}^D \sigma_{j,a} c_a,
    \end{equation}
then we must have $D\ge \varphi(N)$.  

Proposition \ref{prop:lin-ind} is a well-established piece of mathematics.  Before we explain why it is true, we need to define the Euler totient function!  It counts the numbers smaller than $N$ which are relatively prime to $N$: \begin{equation}
    \varphi(N) := |\lbrace k\in \lbrace 1,2,\ldots, N-1\rbrace  : \gcd(k,N)=1\rbrace|. \label{eq:eulertotient}
\end{equation}

\begin{proof}[Proof of Proposition \ref{prop:lin-ind}]
    $\Omega_N$ are the roots of the polynomial 
    \begin{equation}
    P_N(z)
    = z^N-1
    = \prod_{d\mid N}\Phi_d(z),
    \end{equation}
    where the cyclotomic polynomial $\Phi_d(z)$ is defined as \begin{equation}
        \Phi_d(z) := \prod_{k : \gcd(k,d)=1} \left(z-\mathrm{e}^{2\pi \mathrm{i}k/d}\right).
    \end{equation}
    A non-trivial, but standard ~\cite{hardywright2008introduction}, fact is that each $\Phi_d(z)$ has integer coefficients.   We conclude that $\omega := \mathrm{e}^{2\pi \mathrm{i}/N}$ is a root of $\Phi_N(z)$, a polynomial with degree $\varphi(N)$ by construction. Moreover, $\Phi_N$ is irreducible over $\mathbb{Q}$ and $\deg\Phi_N=\varphi(N).$ It follows that$\left[\mathbb{Q}(\omega):\mathbb{Q}\right]=\varphi(N).$  In other words, $\omega^{\varphi(N)}$ is linearly dependent to $1,\ldots,\omega^{\varphi(N)-1}$ in $\mathbb{Z}$.   Repeatedly using this identity, we can express any $\omega^j$ for $j\ge \varphi(N)$ in terms of a linear combination of $1,\ldots,\omega^{\varphi(N)-1}$ with integer coefficients.
\end{proof}

The last sentence of the proof above also allows us to establish the following useful corollary via a direct proof by contradiction:

\begin{cor}\label{cor:cyclotomic}
    If $P(z)$ is a nonzero polynomial of order $p$ with integer coefficients, and $P(\mathrm{e}^{2\pi \mathrm{i}/N}) = 0$, then $p\ge \varphi(N)$.
\end{cor}

\subsection{Proof of Theorem \ref{thm:bound}}
\label{subsec:cyclotomic-fixed-accuracy}
We prove the result by contradiction.  Let $\omega := \mathrm{e}^{2\pi \mathrm{i}/N}$.  Suppose that we have achieved \eqref{eq:goal} using $M<\varphi(N)$ displacement gates.  As per the discussion at the beginning of the section, we  would find that \begin{equation}
    \left(\begin{array}{cc} 1 &\ 0 \end{array}\right) U \left[|0\rangle \otimes  \left(\begin{array}{c} 1 \\ 0 \end{array}\right)\right] = \sum_{(\sigma_1,\ldots,\sigma_M)\in \lbrace \pm 1 \rbrace } c_\sigma |\beta_\sigma\rangle,
\end{equation} where the complex numbers \begin{equation}
    \beta_\sigma := \sum_{j=1}^M \sigma_j \beta_j.
\end{equation}
Here $\beta_j$ are the displacement parameters that appeared in the circuit $U$. For a qubit-dependent displacement, $\sigma_j\in{\pm1}$, while for an unconditional displacement, $\sigma_j$ is fixed to $+1$ on every nonzero path; they need not be independent.   

We first show that if \eqref{eq:goal} is obeyed, then at least one $\beta_\sigma$ must be sufficiently close to each $\alpha \omega^j$ for $j=0,1,\ldots, N-1$.  After all, we need \begin{equation}
   \left|\langle \alpha \omega^j| \left[|\alpha_N\rangle - \sum_{\sigma}c_\sigma |\beta_\sigma\rangle \right]\right|^2 <  \epsilon < \frac{1}{2N} .
\end{equation}
Using \eqref{eq:coherent-overlap}, we see that \begin{equation}
   | \langle \alpha \omega^j|\alpha_N\rangle | > \frac{1}{\sqrt{\mathcal{Z}_{N,\alpha}}} - \frac{1}{\sqrt{\mathcal{Z}_{N,\alpha}}}\sum_{j=1}^{N-1} \exp\left(-\alpha^2 \left(1-\cos \frac{2\pi j}{N}\right)\right). \label{eq:alpha2-1}
\end{equation}
The normalization constant satisfies
\begin{equation}
N\left(1-\sum_{m=1}^{N-1}
\exp\left[
-|\alpha|^2
\left(
1-\cos\frac{2\pi m}{N}
\right)
\right]\right)
\leq
\mathcal Z_{N,\alpha}
\leq
N\left(1+\sum_{m=1}^{N-1}
\exp\left[
-|\alpha|^2
\left(
1-\cos\frac{2\pi m}{N}
\right)
\right]\right).
\label{eq:normalization-constant-two-sided-bound}
\end{equation}
Since $\sum_{m=1}^{N-1}
\exp\left[
-|\alpha|^2
\left(
1-\cos\frac{2\pi m}{N}
\right)
\right]$ tends to zero as $|\alpha|$ tends to infinity,
we may increase the sufficiently-large-$|\alpha|$ threshold so that
\begin{equation}
0.99N
\leq
\mathcal Z_{N,\alpha}
\leq
1.01N,
\label{eq:normalization-constant-numerical-bound}
\end{equation} 
Since each $|c_\sigma|\le 1$ and the number of $\beta_\sigma$ is at most $2^M$, we have the extremely crude bound: \begin{equation}
   \left| \sum_\sigma c_\sigma \langle \alpha \omega^j| \beta_\sigma\rangle \right| < 2^M \max_\sigma \exp\left(-\frac{1}{2}\left|\alpha \omega^j - \beta_\sigma \right|^2\right).
\end{equation}
Suppose every $\beta_\sigma$ satisfies 
\begin{equation}
\left|\alpha \omega^j-\beta_\sigma\right|^2 \geq 20 \log N+2 M \log 2 . \label{eq:minsigma}
\end{equation} 
Therefore, \begin{equation}
    \left| \sum_\sigma c_\sigma \langle \alpha \omega^j| \beta_\sigma\rangle \right| < \frac{1}{N^{10}}. \label{eq:alpha2-2}
\end{equation}
So long as $\alpha>10N^2$, for example, we further have \begin{equation}
    \sum_{j=1}^{N-1}\exp\left(-|\alpha|^2 \left(1-\cos \frac{2\pi j}{N}\right)\right) < 2\sum_{j=1}^{\infty}\exp\left(-\frac{|\alpha|^2}{10} \left(\frac{2\pi j}{N}\right)^2 \right) < \frac{2}{\mathrm{e}^{N^2}-1}. \label{eq:alpha2-3}
\end{equation}
We can readily see that \eqref{eq:alpha2-1}, \eqref{eq:alpha2-2} and \eqref{eq:alpha2-3} are inconsistent for every $N\ge 2$, since $2/(\mathrm{e}^4-1) < 0.04$.  So \eqref{eq:minsigma} must be false and we need to find some $\beta_\sigma$ close to $\alpha \omega^j$ for every $j$: \begin{equation}
\min_\sigma \left|\alpha \omega^j-\beta_\sigma\right|^2 < 20 \log N+2 M \log 2 =: \eta^2 . \label{eq:minsigma2}
\end{equation}

Since \eqref{eq:minsigma2} must hold for every $j$, let us denote with $\sigma_{(j)}$ one such special choice of $\sigma$ for each $j\in \lbrace 0,1,\ldots, N-1\rbrace$, and let \begin{equation}
    \eta_j := \alpha \omega^j - \beta_{\sigma_{(j)}}.
\label{eq:defination-eta}
\end{equation}
From \eqref{eq:minsigma2}, $|\eta_j| \le \eta$.  Define the $M\times N$ matrix $L$ whose columns are $\sigma_{(j)}$ vectors: \begin{equation}
    L := \left(\begin{array}{cccc} \sigma_{(0)} &\ \sigma_{(1)} &\ \cdots &\ \sigma_{(N-1)} \end{array}\right).
\end{equation}
Define the $\mathbb{Q}$-linear map $T:
\mathbb{Q}^{N}
\longrightarrow
\mathbb{Q}(\omega)$ by
\begin{equation}
T(v)
:=
\sum_{j=0}^{N-1}
v_j\omega^j.
\label{eq:cyclotomic-linear-map}
\end{equation}
By Proposition~\ref{prop:lin-ind} we get
\begin{equation}
\dim_{\mathbb{Q}}\ker T=N-\varphi(N),
\qquad
\dim_{\mathbb{Q}}\ker L\geq N-M.
\label{eq:kernel-dimension}
\end{equation}
Because $M<\varphi(N)$, we have
\begin{equation}
\dim_{\mathbb{Q}}
\ker L
\geq
N-M
>
N-\varphi(N)
=
\dim_{\mathbb{Q}}
\ker T.
\label{eq:kernel-dimension-comparison}
\end{equation}
Hence, for every possible sign matrix $L$, there exists a vector
$v_L\in\mathbb{Q}^{N}$ such that
\begin{equation}
Lv_L
=
0,
\qquad
T(v_L)
\neq
0.
\label{eq:choice-vL}
\end{equation}
We normalize this vector so that
\begin{equation}
\left\|
v_L
\right\|_1
=
\sum_{j=0}^{N-1}
\left|
(v_L)_j
\right|
=
1.
\label{eq:vL-normalization}
\end{equation}

There are only finitely many sign matrices with fewer than
$\varphi(N)$ rows and $N$ columns. We may therefore define
\begin{equation}
u_N
:=
\min_L
\left|
T(v_L)
\right|
>
0,
\label{eq:ustar}
\end{equation}
where the minimum is taken over all possible sign matrices $L$ with
$M<\varphi(N)$. If \eqref{eq:defination-eta} holds,
\begin{equation}
    \alpha T(v_L)
    =
    \sum_{j=0}^{N-1}(v_L)_j\beta_{\sigma_{(j)}}+\sum_{j=0}^{N-1}(v_L)_j\eta_j =
    \beta^{\mathsf{T}}Lv_L  +\sum_{j=0}^{N-1}(v_L)_j\eta_j.
\end{equation}
Since $Lv_L=0$, the first sum vanishes and using the triangle inequality, we get
\begin{equation}
    \alpha T(v_L)=\sum_{j=0}^{N-1}(v_L)_j\eta_j\leq\eta\sum_{j=0}^{N-1}|(v_L)_j|
\end{equation}
Combining with \eqref{eq:ustar}, we could get
\begin{equation}
    \alpha u_N\leq\eta.
\end{equation}
Since $M<\varphi(N)$, we may replace $\eta$ by the uniform upper bound $\eta_N^2=20\log N+2\bigl(\varphi(N)-1\bigr)\log2$, therefore \begin{equation}
    \alpha > \max\left(10N^2, \frac{ \eta_N}{u_N}\right) \label{eq:alphasufficientlarge}
\end{equation} leads to the desired contradiction.  It cannot be possible to make an approximate $N$-fold cat state whenever \eqref{eq:alphasufficientlarge} is obeyed.

\subsection{Bounding the runtime}
We end this section with another very simple result, bounding the runtime needed to (approximately) create $|\alpha_N\rangle$ with the phase space instruction set.

\begin{prop}\label{prop:runtimebound}
For any circuit using the phase space instruction set that prepares $|\alpha_N\rangle$ with fidelity $>\frac{2}{3}$ from a single coherent state $|\beta\rangle$, for any $\beta \in \mathbb{C}$, \begin{equation} \label{eq:trun-alpha}
t_{\mathrm{run}}
\geq
\mathrm{\Omega}(|\alpha|)
\end{equation} for suitably large $\alpha$.  The same conclusions hold for any cat state with arbitrary phases: \begin{equation}
    |\alpha_N^{\text{phases}}\rangle := \frac{1}{\sqrt{\mathcal{Z}}} \sum_{j=0}^{N-1} \mathrm{e}^{\mathrm{i}\gamma_j} \left| \alpha \mathrm{e}^{2\pi \mathrm{i}j/N}\right\rangle . \label{eq:cat-arbitrary-phases}
\end{equation}
\end{prop}
\begin{proof}
 The only gate in the phase space instruction set that can move apart two coherent states is the spin-dependent displacement $D(z X_n)$. If we apply $M$ of these spin-dependent displacement gates for a total runtime of $t_{\mathrm{run}}$, then the most general action of our unitary is \begin{equation}
        U \left(\begin{array}{c} |\beta\rangle \\ 0\end{array}\right) = \sum_{\sigma_1,\ldots, \sigma_M = \pm 1} \left(\begin{array}{c} c^\uparrow_{\sigma_1\cdots \sigma_M} \\ c^\downarrow_{\sigma_1\cdots \sigma_M}\end{array}\right) \left| \beta + \sigma_1 z_1 + \cdots + \sigma_M z_M \right\rangle 
    \end{equation}
    where \begin{equation}
        \sum_{j=1}^M |z_j| \le t_{\text{run}}. \label{eq:zjtrun}
    \end{equation}
    Let
    \begin{equation}
    L
    :=
    \sum_{j=1}^M|z_j|.
    \end{equation}
    Consider two target peaks with centers
$\alpha$ and $\alpha\mathrm{e}^{\mathrm{i}\theta}$, where $\theta\in[\pi/3,5\pi/3]$. Since $\cos\theta\leq 1/2$ on this interval, their separation satisfies
\begin{equation}
\left|
\alpha\mathrm{e}^{\mathrm{i}\theta}-\alpha
\right|^2=2|\alpha|^2\left(1-\cos\theta\right)
\nonumber\geq
|\alpha|^2.
\end{equation}
Thus, these two target coherent-state peaks are separated by at least
$|\alpha|$. For fixed
$N$, the target coherent-state peaks become mutually orthogonal as
$|\alpha|\rightarrow\infty$. A state having non-negligible
contributions near at most half of the target peaks can therefore have
fidelity at most
\begin{equation}
\mathcal{F}
\leq
\frac{1}{2}
+
\mathcal{O}_N\left(
\mathrm{e}^{-c_N|\alpha|^2}
\right)
<
\frac{2}{3}
\end{equation}
for sufficiently large $|\alpha|$. Hence,  preparing the target
cat state with fidelity greater than $2/3$ requires the output to
reproduce two such separated peaks. Let the corresponding generated
coherent-state centers be
\begin{equation}
\gamma_{\sigma}
=
\beta+\sum_{j=1}^{M}\sigma_j z_j,
\qquad
\gamma_{\tau}
=
\beta+\sum_{j=1}^{M}\tau_j z_j.
\end{equation}
Taking their difference eliminates the initial center $\beta$:
\begin{equation}
\gamma_{\sigma}-\gamma_{\tau}
=
\sum_{j=1}^{M}
\left(\sigma_j-\tau_j\right)z_j.
\end{equation}
By the triangle inequality,
\begin{equation}
\left|
\gamma_{\sigma}-\gamma_{\tau}
\right|
\leq
\sum_{j=1}^{M}
\left|
\sigma_j-\tau_j
\right|
|z_j|
\nonumber\leq
2\sum_{j=1}^{M}|z_j|
\nonumber
=
2L.
\end{equation}
Since the two target peaks are separated by a distance of order
$|\alpha|$, reproducing both peaks requires
\begin{equation}
2L
=
\mathrm{\Omega}(|\alpha|).
\end{equation}
Combining this with Eq.~\eqref{eq:zjtrun}, we obtain \eqref{eq:trun-alpha}.
\end{proof}

\section{Protocol for preparing $N$-fold cats with prime $N$}
\label{sec:4}
In this section, we prove the existence of a saturating protocol that can achieve \eqref{eq:goal}, at least for sufficiently large $\alpha$, in the case where $N$ is prime.   We remind the reader that when $N$ is prime, $\varphi(N) = N-1$.  Therefore, our goal will be to establish the following result: \begin{thm}\label{thm:protocol}
    For any prime $N$, there exists a number $\alpha_{\mathrm{up}}^*(N,\epsilon)$ such that if $|\alpha|>\alpha_{\mathrm{up}}^*(N,\epsilon)$, there exists a protocol using a circuit of depth $4N$ consisting of gates in $\mathcal{S}$, that prepares $|\alpha_N\rangle$ with squared state-vector $1-\epsilon$, i.e. achieves \eqref{eq:goal}.   Moreover, for fixed $N$ and $\epsilon$,  this protocol runs in time \begin{equation} \label{eq:trun-good}
        t_{\mathrm{run}} = \mathrm{\Theta}(\alpha).
    \end{equation}
    The protocol can be generalized to prepare cat states with arbitrary phases, as in \eqref{eq:cat-arbitrary-phases}
\end{thm}
\noindent The proof of this result is constructive, but the description of the protocol is somewhat involved.  Crucially, for prime $N$, this protocol is asymptotically optimal.  It saturates both the circuit depth bound from Theorem \ref{thm:bound} (as $\varphi(N)=N-1$ when $N$ is prime), and the runtime bound from Proposition \ref{prop:runtimebound}.  We will present a sketch of how the protocol works in the next subsection.

The optimal scaling of $t_{\text{run}}$ is an important feature of our construction.
Although the cyclotomic obstruction forces many sequential
displacement instructions, the protocol avoids the naive
multiplicative cost that would result from paying the full
phase-space radius at every step. For prime $N$, the protocol matches the established linear depth scaling while remaining economical in its total phase-space motion. It is therefore both asymptotically depth-efficient and geometrically efficient within the control model considered here. 
\begin{rmk}[Order of limits]
\label{rmk:order-of-limits}
Both principal theorems are pointwise in $N$: for every fixed $N$
and fixed target error, the result holds once $|\alpha|$ exceeds an
$N$-dependent threshold. The present proofs do not provide uniform
bounds on these thresholds as $N$ increases. Thus, the comparison of
the lower and upper $N$-scalings is understood along sequences
satisfying $N\rightarrow \infty$ and 
\begin{equation}
|\alpha(N)|
\geq
\max\left\{
\alpha_{\mathrm{low}}^{*}(N,\epsilon),
\alpha_{\mathrm{up}}^{*}(N,\epsilon)
\right\}.
\label{eq:joint-N-alpha-regime}
\end{equation}
\end{rmk}

\subsection{Overview of the protocol}

We first specify the ideal operation we would like to implement. Let $\{a_{j}\}_{j=1}^{n}\subset\mathbb{C}$ be distinct coherent-state parameters, and let $a_{*}\in\mathbb{C}$ be a distinguished point in phase space. We consider input states of the form
\begin{equation}
\ket{\Psi_{\mathrm{in}}} =
\begin{pmatrix}
\sum_{j}c_{j}\lvert a_{j}\rangle \\
c_{*}\lvert 0\rangle
\end{pmatrix}. \label{eq:Psiin}
\end{equation}
where the upper component collects a superposition of already-prepared coherent states and the lower component stores amplitude in the vacuum $\lvert 0\rangle$. 
\begin{defn}[Ideal $\mathcal{V}$ gate] \label{def:idealV}
For a fixed real parameter $\kappa\in(0,1]$, define the \emph{ideal} linear map $\mathcal{V}_{\mathrm{ideal}}$ on such inputs by
\begin{equation}
\label{eq:ideal-action}
\mathcal{V}_{\mathrm{ideal}}
\begin{pmatrix}
\sum_{j}c_{j}\lvert a_{j}\rangle \\
c_{*}\lvert 0\rangle
\end{pmatrix}
:=
\begin{pmatrix}
\sum_{j}c_{j}\lvert a_{j}\rangle + \kappa c_{*}\lvert a_{*}\rangle \\
\sqrt{1-\kappa^{2}}\,c_{*}\lvert 0\rangle
\end{pmatrix}.
\end{equation}
\end{defn}

\noindent In words, $\mathcal{V}_{\mathrm{ideal}}$ takes a fraction $\kappa$ of the amplitude $c_{*}$ initially stored in the vacuum, converts it to the new coherent state $\lvert a_{*}\rangle$ in the upper component, and leaves the remaining amplitude $\sqrt{1-\kappa^{2}}\,c_{*}$ in the vacuum. This is exactly the primitive we need in order to ``grow'' an $N$-fold cat state by successively adding coherent components, as we depict in Figure \ref{fig:protocol}. The initialization of the protocol and the explicit choice of the
amplitude-transfer parameters $\kappa_k$ are given in Appendix~\ref{app:initialization-schedule}.

\begin{figure}[t]
    \centering
    \includegraphics[width=1\textwidth]{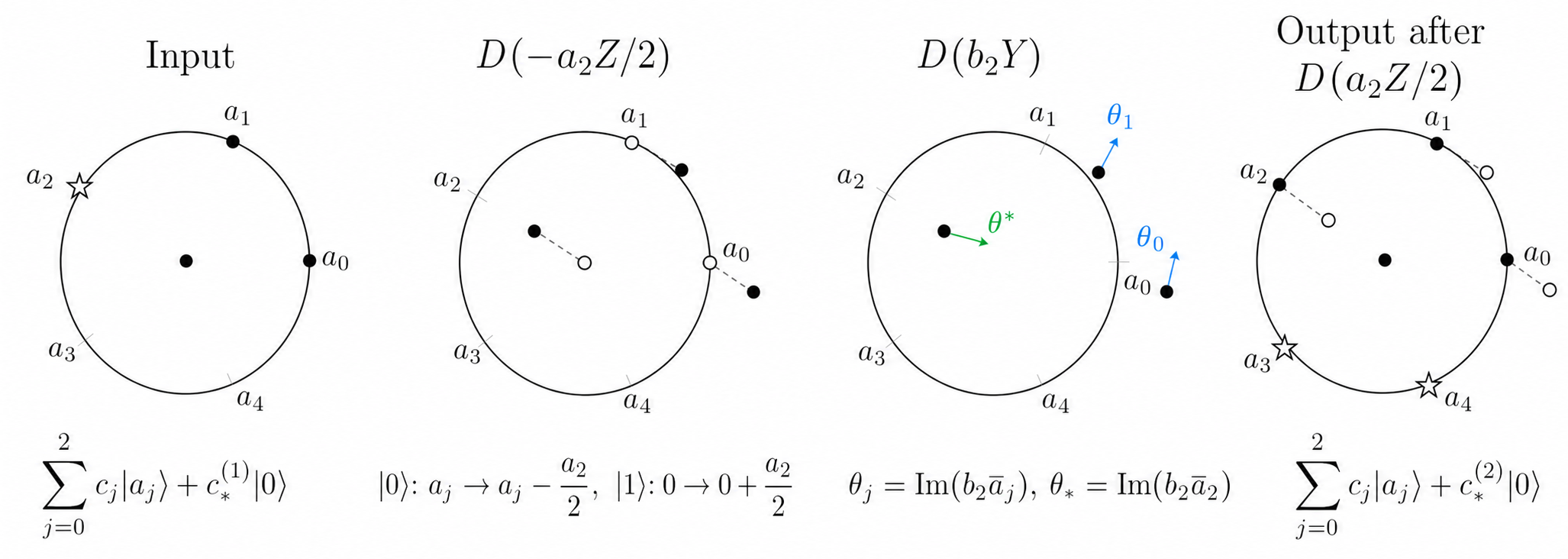}
    \caption{
    One step in the \(N=5\) cat-state protocol, illustrated in phase space. 
    From left to right: 
    (i) the conditional displacement \(D(-a_{2}Z/2)\), which shifts the already-prepared peaks on the \(\ket{0}\) branch by \(-a_{2}/2\) and the lower-branch vacuum component by \(+a_{2}/2\); 
    (ii) the Y-controlled displacement \(D(b_{2}Y)\), which generates the phase responses \(\theta_{0}\), \(\theta_{1}\), and \(\theta_{*}\); 
    (iii) the qubit rotation $\mathrm{e}^{\mathrm{i}\theta_2 Y}$, which rotates the common ancilla direction to the computational basis; and 
    (iv) the final conditional displacement \(D(a_{2}Z/2)\), which returns the state to the original circle while adding the new coherent component \(a_{2}\). 
    Filled dots denote already occupied coherent components, stars denote new or future target peaks, and crosses mark the ideal circle locations.  
    }
    \label{fig:protocol}
\end{figure}

We now define the concrete gate $V$ that we can implement physically, and that is designed to approximate $\mathcal{V}_{\mathrm{ideal}}$ on the above subspace. The gate consists of three conditional displacements on the oscillator, controlled by the qubit, together with an intermediate single-qubit rotation.

\begin{defn}[Actual $V$ gate]
\label{def:V-operator}
Fix a complex number $a_{*}$ and parameters $\lambda>0$, $\kappa,\phi\in\mathbb{R}$. Define
\begin{equation}
\label{eq:V-operator}
V(a_{*},\lambda \mathrm{e}^{\mathrm{i}\phi},\kappa)
:= D\!\left(\frac{a_{*}}{2}Z\right)\,
   \mathrm{e}^{\mathrm{i}\theta Y}\,
   D\!\left(\lambda \mathrm{e}^{\mathrm{i}\phi} Y\right)\,
   D\!\left(-\frac{a_{*}}{2}Z\right).
\end{equation}
 The dependence on $\kappa$ enters through the choice of the parameter $\lambda \mathrm{e}^{\mathrm{i}\phi}$, which will be chosen so that $V(a_{*},\lambda \mathrm{e}^{\mathrm{i}\phi},\kappa)$ approximates the ideal action \eqref{eq:ideal-action}.
\end{defn}

\noindent Our goal is to choose specific values of the parameters
$a_{*}$, $\lambda$, $\phi$, and $\theta$ such that
$V(a_{*},\lambda \mathrm{e}^{\mathrm{i}\phi},\kappa)$ approximates the ideal transformation
$\mathcal{V}_{\mathrm{ideal}}$ to accuracy $h$. To achieve this, it is not enough to tune only the gate parameters: the input state $\ket{\Psi_{\mathrm{in}}}$ must also be restricted to take a special form \eqref{eq:Psiin} with particular choices of $a_j$.  In the preparation of the $N$-fold cat state, these $a_j$ will correspond to a subset of the desired pieces of the final cat state; moreover, the values of $a_j$ must be chosen very carefully to enable the depth-4 circuit $V$ to actually grow a new leg of the cat.  Constructing such $a_j$ is a non-trivial part of the proof of Theorem \ref{thm:protocol}.  We will show how to choose $a_j$ and $V$ such that
\begin{equation}
\label{eq:V-approx-ideal}
\big\lVert \bigl(\mathcal{V}_{\mathrm{ideal}} - V(a_{*},\lambda \mathrm{e}^{\mathrm{i}\phi},\kappa)\bigr)\,\ket{\Psi_{\mathrm{in}}}\big\rVert
\leq h
\end{equation}
is achievable.  If the loss in fidelity after each $\mathcal{V}_{\text{ideal}}$ is $h$, and we must apply $N$ different $\mathcal{V}_{\text{ideal}}$ gates, then we can establish Theorem \ref{thm:protocol} by choosing $h < N^{-1}\epsilon$. 

The key thing needed to prove Theorem \ref{thm:protocol} is, therefore, the prescription for how to choose the physical gates $V$ to ensure \eqref{eq:V-approx-ideal}.  We will again rely on ideas from mathematics --  Diophantine approximability and ergodic theory -- to verify these gates exist.  Similar to the proof of Theorem \ref{thm:bound}, a key idea will be that the independence of coherent state positions in the complex plane enables us to ``independently" control the spin orientation at each coherent state peak at suitable values of $\alpha$.   The proof of the theorem is admittedly rather tedious, and we present it in Appendix \ref{app:proof}.

Finally, we remark that although the actual $V$ gate in \eqref{eq:V-operator} involves two large displacements, we find that these displacements can be combined between adjacent steps to form a much smaller displacement gate, which gives us the asymptotically optimal runtime \eqref{eq:trun-good}.

\subsection{Protocol performance at small $\alpha$}
\label{sec:numerical-results}
Although Theorem \ref{thm:protocol} is only a formal guarantee of a high-fidelity state preparation scheme at very large $\alpha$, it is worthwhile asking how well our formal protocol works at pretty small values $\alpha \sim 10$, which is well below the threshold needed by \eqref{eq:alphasufficientlarge} (we expect a similar concern holds for the minimum required $\alpha$ for Theorem \ref{thm:protocol}).   This is because in experiments it is difficult to prepare cat states with coherent state parameter $\alpha$ too much larger than this.  For simplicity in what follows, we asume that the parameter $\alpha$ is real and positive. We will largely focus on preparing the 5-fold cat states $|\alpha_5\rangle$, as we do not know of any previously-constructed  protocols to prepare such a cat state with the phase space instruction set.  

We compare three levels of control: the unoptimized analytic protocol, a gate-level refinement of that protocol, and a pulse-level GRAPE optimization~\cite{khaneja2005optimal}. For a joint oscillator--qubit output state $\ket{\Psi_{\mathrm{out}}}$, we define the target fidelity by
\begin{equation}
\mathcal{F}
:=
\left|
\left\langle
\Psi_{\mathrm{target}}
\middle|
\Psi_{\mathrm{out}}
\right\rangle
\right|^2,
\qquad
\ket{\Psi_{\mathrm{target}}}
:=
\begin{pmatrix}
\ket{\alpha_5}\\
0
\end{pmatrix}.
\label{eq:numerical-fidelity}
\end{equation}

We first evaluate two explicit analytic protocols for $N=5$. Their scaled displacement parameters are fixed, while the corresponding physical $Y$-controlled displacement amplitudes scale as $1/\alpha$. The parameters are listed in Appendix~\ref{app:N5}. The two protocols are obtained using Diophantine tolerances $\varepsilon_{\mathrm{ph}}=0.029$ and $\varepsilon_{\mathrm{ph}}=0.065$. Figure~\ref{fig:protocol-direct-fidelity} shows the fidelity obtained directly from these protocols. In both cases, the fidelity increases with $\alpha$. The protocol constructed using the looser Diophantine tolerance performs better at moderate $\alpha$, showing that a tighter asymptotic phase-alignment condition does not necessarily give better finite-$\alpha$ performance.

\begin{figure}[t]
    \centering
    \includegraphics[width=0.7\textwidth]{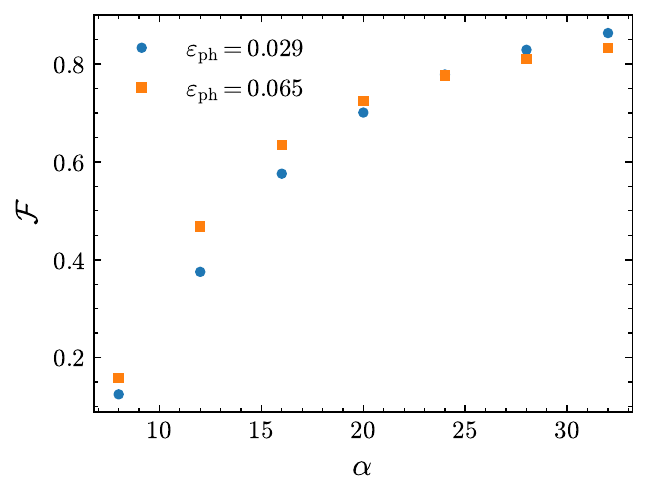}
    \caption{
    Fidelity obtained directly from the protocol as a function of the cat
    radius \(\alpha\).  The two curves correspond to two different choices
    of the Diophantine condition, labelled by \(\varepsilon_{\mathrm{ph}}=0.029\) and
    \(\varepsilon_{\mathrm{ph}}=0.065\).  In both cases, the fidelity increases as
    \(\alpha\) becomes larger. 
    }
    \label{fig:protocol-direct-fidelity}
\end{figure}

We next use the better-performing analytic sequence as a structured optimization seed at $\alpha=12$, and then use a cutoff-free gate-level refinement, which keeps the gate ordering fixed while optimizing the real and imaginary parts of the displacement amplitudes and all qubit-rotation angles.  We
maximize the joint target fidelity $\mathcal F$ defined above. The gradients are computed by automatic differentiation, and the optimization
is performed with L-BFGS-B, a limited-memory quasi-Newton method that approximates the inverse Hessian while enforcing box constraints on the variables~\cite{byrd1995lbfgsb}. To allow the search to move gradually away from
the protocol seed, we use a sequence of increasingly broad optimization
windows for the \(D\) and \(D_Z\) amplitudes. The circuit is evaluated directly in a finite coherent-branch representation, without introducing a Fock-space cutoff. The unoptimized protocol has fidelity $\mathcal{F}=0.468.$ After gate-level refinement, the fidelity increases to $\mathcal{F}=0.595$ as is shown in Fig.~\ref{fig:alpha12-wigner}.  The optimized sequence contains \(512\) merged coherent branches.  Thus, it improves the protocol seed while avoiding any
truncation of the oscillator Hilbert space. 

We then translate the optimized gate sequence into piecewise-constant
Hamiltonian controls.  Working in the rotating frame used for the pulse
simulation, the Hamiltonian is
\begin{equation}
\begin{aligned}
H(t)
={}&c_x(t) I\otimes X
+c_y(t) I\otimes Y
+c_z(t) I\otimes Z +g_{xZ}(t)x\otimes Z
+g_{pZ}(t)p\otimes Z \\
&+u_{xI}(t)x\otimes I
+u_{pI}(t)p\otimes I .
\end{aligned}
\label{eq:grape-hamiltonian-numerics}
\end{equation}
For a time slot of duration \(\mathrm{\Delta} t\), the gate-to-pulse conversion is
\begin{subequations}\begin{align}
D(\beta):
\qquad&
u_{xI}
=
-\frac{\sqrt{2}\,\mathrm{Im}(\beta)}{\mathrm{\Delta} t},
\qquad
u_{pI}
=
\frac{\sqrt{2}\,\mathrm{Re}(\beta)}{\mathrm{\Delta} t},
\\
D_Z(\beta):
\qquad&
g_{xZ}
=
-\frac{\sqrt{2}\,\mathrm{Im}(\beta)}{\mathrm{\Delta} t},
\qquad
g_{pZ}
=
\frac{\sqrt{2}\,\mathrm{Re}(\beta)}{\mathrm{\Delta} t},
\\
R(\theta):
\qquad&
c_y
=
-\frac{\theta}{\mathrm{\Delta} t}.
\end{align}\end{subequations}
The \(D_Y\) gates are implemented by conjugating a \(D_Z\) pulse with two
qubit \(X\)-rotations.  Consequently, the 18-gate 
is converted into a pulse with 26 piecewise-constant time slots and seven
independent control channels. 
The pulse-level simulation uses a Fock cutoff \(n_{\mathrm{cut}}=600\), so
the joint Hilbert-space dimension is \(1200\). We choose this cutoff
because the calculated fidelity is nearly converged at
$n_{\mathrm{cut}}=600$ for the parameter regime considered here. We take the initialized ancilla
superposition as the input state in order to omit one fixed
single-qubit gate.  Hence our initial state is:
\begin{equation}
\ket{\Psi_{\mathrm{in}}}
=
\left(\begin{array}{c} \sqrt{\frac{1}{5}}\ket{0} \\ \sqrt{\frac{4}{5}}\ket{0} \end{array}\right),
\end{equation}
and the target is
\begin{equation}
\ket{\Psi_{\mathrm{target}}}
=
\left(\begin{array}{c} \ket{12_5} \\ 0 \end{array}\right).
\end{equation}

Starting from this pulse, we perform one continuous GRAPE optimization using
L-BFGS-B with all seven control channels allowed to vary.  The control
amplitudes are bounded in the interval \([-25,25]\), and the optimization is
run for 80 iterations.  The final fidelity reaches $\mathcal{F} = 0.71684$. 
Thus, for \(\alpha=12\), the full protocol-seeded optimization pipeline
raises the fidelity from approximately \(0.47\) for the initial protocol gate
sequence to approximately \(0.59\) after gate-level refinement,
and finally to approximately \(0.72\) after pulse-level GRAPE refinement as shown in Fig.~\ref{fig:alpha12-controls} and Fig.~\ref{fig:alpha12-wigner}.

\begin{figure}[t]
    \centering
    \includegraphics[width=0.75\textwidth]{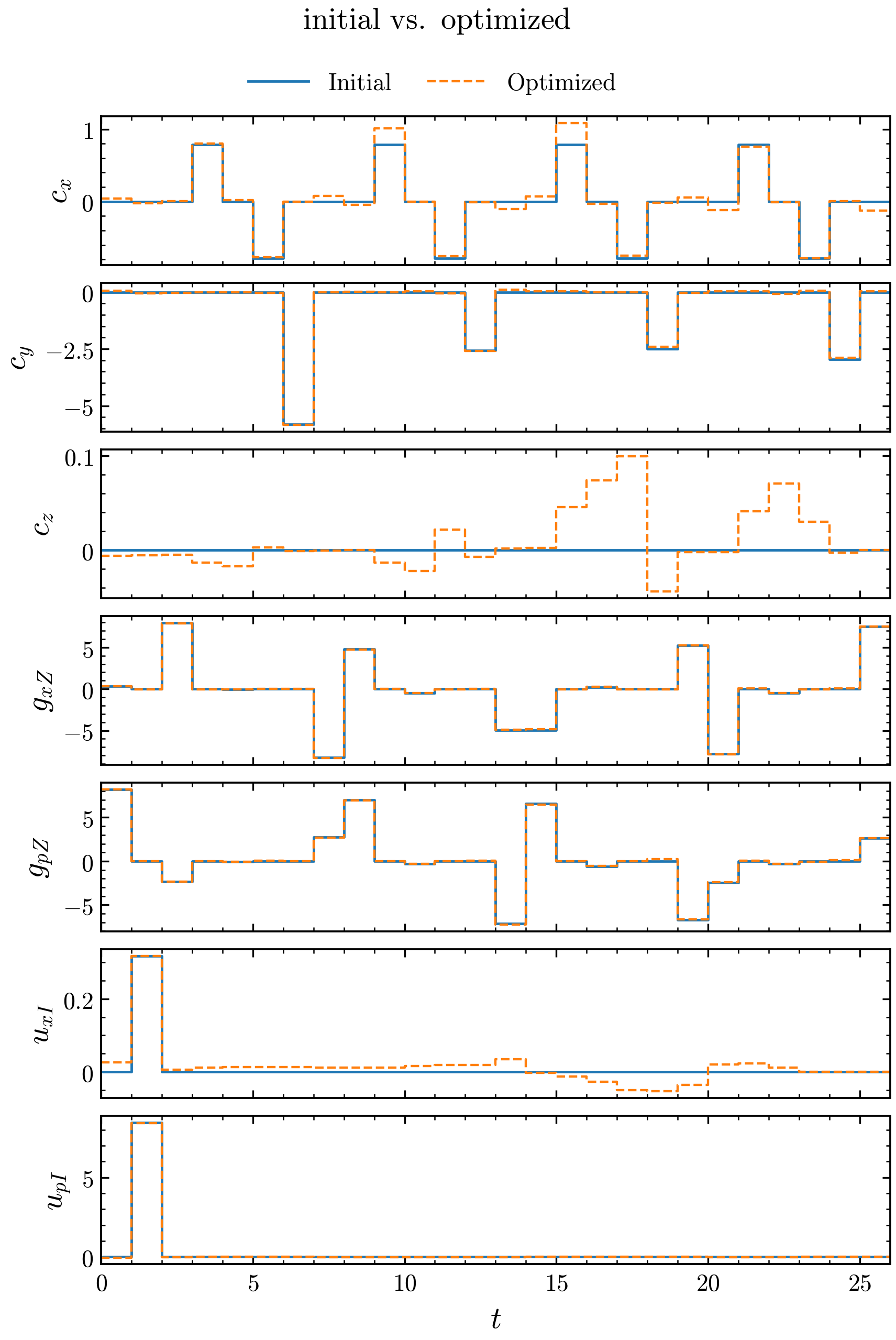}
    \caption{
    Initial pulse obtained by translating the optimized coherent-branch gate
    sequence, shown as solid blue curves, and the pulse after 80 GRAPE
    iterations, shown as dashed orange curves.  The results are for the
    five-component cat state with \(\alpha=12\).
    }
    \label{fig:alpha12-controls}
\end{figure}

\begin{figure}[t]
    \centering
    \includegraphics[width=1\textwidth]{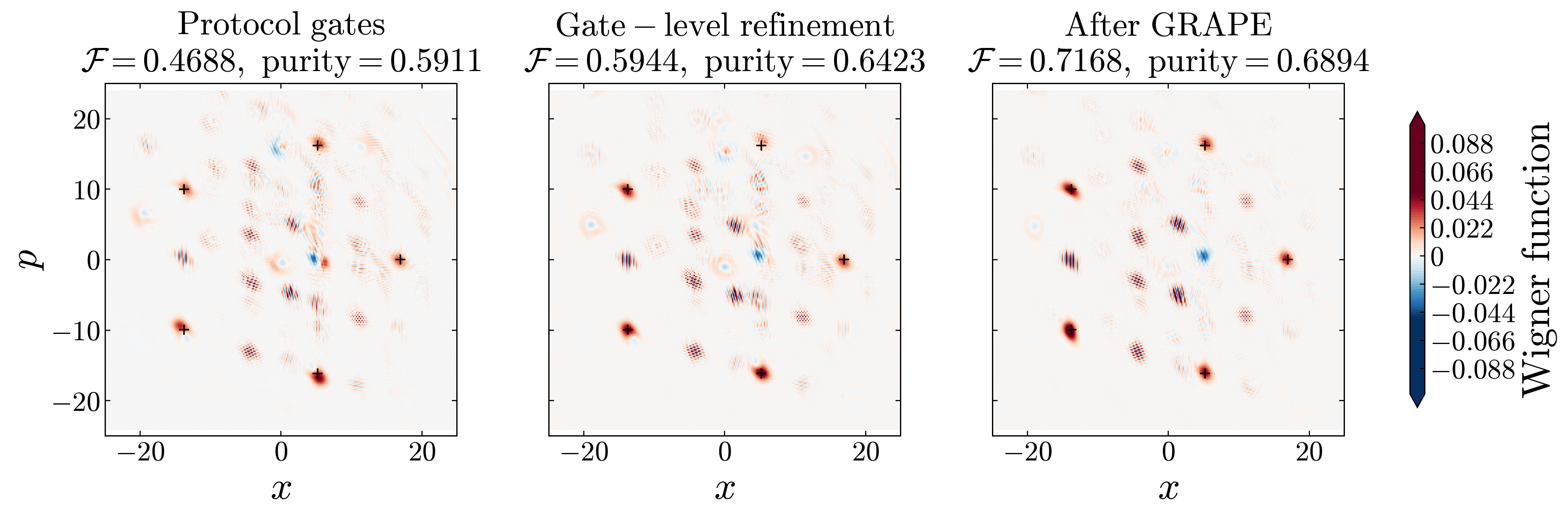}
    \caption{
Reduced oscillator Wigner functions for the preparation of the five-component cat state at $\alpha=12$. From left to right: the output of the unoptimized analytic protocol, the result after gate-level refinement, and the result after pulse-level GRAPE optimization. The corresponding fidelities are $\mathcal{F}=0.4688$, $\mathcal{F}=0.5944$, and $\mathcal{F}=0.7168$, while the oscillator purities are $\mathrm{Tr}(\rho_{\mathrm{osc}}^2)=0.5911$, $0.6423$, and $0.6894$, respectively. Black crosses indicate the target coherent-state locations. All panels use the same phase-space range and color scale.
}
    \label{fig:alpha12-wigner}
\end{figure}

Finally, we compare two complete numerical optimization pipelines. In the protocol-seeded pipeline, the analytic protocol is first refined at the gate level, translated into Hamiltonian controls, and then optimized using GRAPE. In the random-seed pipeline, GRAPE begins directly from a randomly initialized pulse. So the total optimization budgets of the two pipelines are not identical. The following results
should therefore be interpreted as a comparison between two complete
optimization pipelines, rather than as a controlled comparison of the
two pulse initializations alone. 

For small \(\alpha\), direct GRAPE from a random initial pulse can perform
comparably or better.  However, after \(\alpha=7\), the protocol-seeded
method becomes superior.  At \(\alpha=7\), the protocol-seeded method reaches
a fidelity of \(0.843\), compared with \(0.809\) for the random pulse.  At
\(\alpha=9\), the two fidelities are \(0.787\) and \(0.667\), respectively.
The difference becomes particularly large at \(\alpha=12\), where the
protocol-seeded method reaches \(0.717\), while direct GRAPE from a random
pulse reaches only \(0.407\), as shown in Figure \ref{fig:protocol-seed-comparison}.  

\begin{figure}[t]
    \centering
    \includegraphics[width=0.49\textwidth]
        {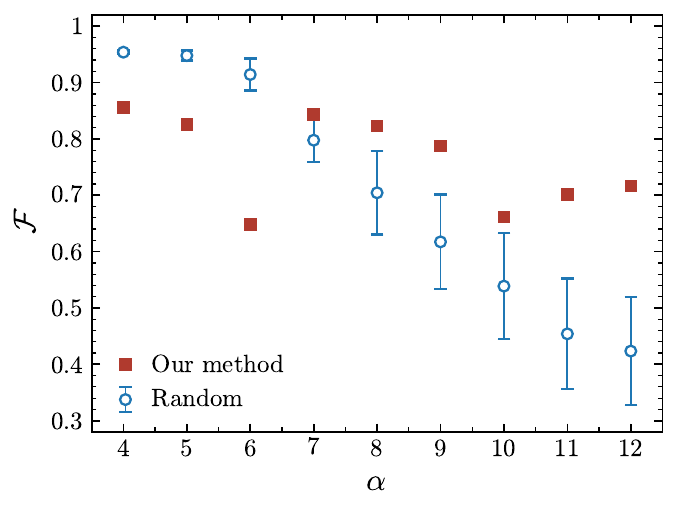}
    \hfill
    \includegraphics[width=0.49\textwidth]
        {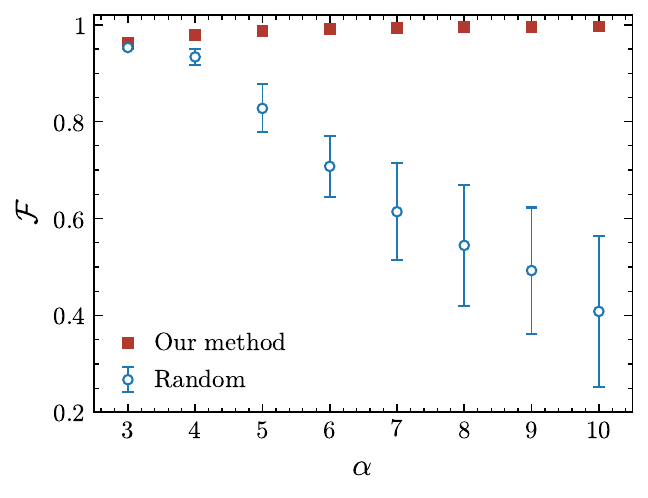}
    \caption{
        Comparison between direct GRAPE optimization from 20 random pulses
        vs. the protocol-seeded pipeline. Both pulse-level calculations use the same number of time slots, pulse duration, seven control channels, amplitude bounds and the same 80-iteration GRAPE
budget. The protocol-seeded pipeline additionally includes a
coherent-branch gate-level optimization, so the total optimization
budgets are not identical. The comparison is therefore between two
complete optimization pipelines. The error bars show the full range across the twenty random initializations. Left: \(N=5\). Right: \(N=3\).
    }
    \label{fig:protocol-seed-comparison}
\end{figure}

Figure \ref{fig:protocol-seed-comparison} also shows the performance of a protocol at $N=3$.  Our method works extremely well here because we can actually solve the ``number-theoretic" equations explicitly: see Appendix \ref{app:N3-exact-solution}.

These results show that our protocol is especially valuable in the
large-\(\alpha\) regime for prime $N>3$, and for the case of $N=3$ at any $\alpha$.  It provides a structured initial condition in a
region of control space that is difficult for a short, blind GRAPE
optimization to find from a random pulse.

\section{Compilation difficulty}
\label{sec:5}
In this section, we combine the results from before to establish the following simple theorem: \begin{thm}
    There exist state preparation protocols for which there is no efficient circuit compilation (as measured by runtime) between the phase-space instruction set, and a universal gate set with access to the gate \begin{equation}
        U_{\text{rot}}(\theta) := \exp\left(\mathrm{i}\theta a^\dagger a Z\right).
    \end{equation}
\end{thm}
\begin{proof}
     There is a recent \cite{fong2025engineering} construction of a unitary $U_0$ which efficiently takes a single coherent state into the cat state: \begin{equation}
         U_0 |\alpha\rangle \approx  \frac{1}{\sqrt{N}}
\sum_{\ell=0}^{N-1}
\exp\left[
\frac{\mathrm{i}\pi\ell(\ell-N)}{N}
\right]
\ket{\alpha \mathrm{e}^{\mathrm{i}2\pi\ell/N}}. \label{eq:cat-special-phases}
     \end{equation} The approach is based on quantum signal processing \cite{Motlagh:2023oqc}.   Since this circuit requires $\mathrm{O}(N)$ single-qubit rotation gates, each of which take O(1) time to implement, the runtime of this circuit is $\mathrm{O}(N)$. 
     
     As remarked in Proposition \ref{prop:runtimebound} and Theorem \ref{thm:protocol}, the specific choice of phases in \eqref{eq:cat-special-phases} do not make the state any easier to prepare for the phase space instruction set.  There is no efficient way to compile operations between these two circuits that preserves the runtime.
\end{proof} 

A more physical picture for this compilation inefficiency result can be given.   Suppose we wish to convert a circuit made out of $2N$ $U_{\text{num}}$ gates into a circuit made out of $\mathcal{S}$ gates.   \eqref{eq:universalcommutator} suggests how we can do that, interweaving together $D(\epsilon X)$ and $D(\epsilon Y)$ to make $U_{\text{num}}(\epsilon^2)$ as in a Trotter expansion.  If the approximate $U_{\text{rot}}$ gate is to act on a state $|\psi\rangle$, the Trotter error in this expansion could scale as $\lVert \epsilon^3 a^3 |\psi\rangle \rVert \sim (\epsilon \alpha)^3$ (this is the next order in the Taylor series beyond the one of interest).  Here we estimate $a\sim \alpha$ as we are acting on a coherent state of size $\alpha$. Since the protocol uses angular steps of size $\theta \sim N^{-1}$, to make the error in the circuit of depth \begin{equation}
    D \sim \frac{\theta N}{\epsilon^2} \sim \frac{1}{\epsilon^2}
\end{equation} very small, we might take \begin{equation}
    1 \gg (\epsilon \alpha)^3 D \sim \alpha^3 \epsilon 
\end{equation}   To implement the circuit at fixed $\theta$ also requires, as $\alpha$ becomes asymptotically large, \begin{equation}
    t_{\text{run}} \sim D\cdot \epsilon \sim \frac{1}{\epsilon} \gg \alpha^3.
\end{equation}

In contrast, we saw from Theorem \ref{thm:protocol} that we can achieve the optimal (for the phase space instruction set) $t_{\text{run}} \sim \alpha$, which is a parametrically better scaling at large $\alpha$.  We conclude that the cross-compilation into QSP may be a suboptimal protocol for the phase space instruction set (although our estimate of the Trotter error is also rather harsh; there may be a more efficient cross-compilation).

\section{Conclusion}
In this work, we studied the preparation of rotationally symmetric $N$-fold cat
states using the phase space instruction set, consisting of single-qubit rotations
and qubit-dependent displacements of a single bosonic mode.  We proved that, in
the large-amplitude regime, any sufficiently accurate preparation protocol must
use at least $\varphi(N)$ independent phase-space instructions, where $\varphi$
is Euler's totient function.  The obstruction comes from the cyclotomic structure
of the $N$th roots of unity.  For prime $N$, this lower bound is linear in $N$.

We also gave an explicit preparation protocol for prime $N$ whose depth scales
linearly with $N$, matching the lower-bound scaling up to constant factors.  The
protocol builds the target state by adding coherent-state components one at a
time, while choosing displacement directions so that previously prepared
components remain aligned and the new component is coherently incorporated.  Thus,
for prime-fold cat states, the lower bound and the construction together identify
the correct asymptotic scaling within the phase space instruction set.  Although being large circuit depth, this protocol has an asymptotically optimal runtime; in this respect, it may still be efficient for experiments if large-displacement gates are slow.

Our story illustrates some subtle features of bosonic control: even though non-Gaussian controls might grow the number of coherent states exponentially fast, it is not trivial to move these desired coherent states to the right locations to form highly non-Gaussian states like $N$-fold cat states.  The circuit complexity of our protocol is beyond any simple bound based simply on quantum speed limits, and comes from the algebraic structure of the target state. Each of these features crucially rely on the large-dimensional Hilbert space.

The result also gives a useful warning for bosonic compilation.  In a finite
qubit system, once a universal gate set is fixed, changing between universal gate
sets often introduces only relatively mild overheads for state preparation.  For
oscillators, the situation can be more delicate.  A multi-component cat state may
be natural to prepare using an instruction set with engineered nonlinear
operations, such as SNAP-type or signal-processing-based controls, while the same
state requires linearly many independent displacement directions in the phase
space instruction set studied here.  Thus, even when two oscillator control
models are both universal, the cost of translating a preparation strategy from
one model to another can reflect the detailed algebraic structure of the target
state.  Understanding such structure-dependent overheads is an important step
toward a more quantitative theory of compilation for continuous-variable quantum
systems.

\addcontentsline{toc}{section}{Acknowledgements}
 \section*{Acknowledgments}
We thank Shraddha Singh and Nathan Wiebe for useful discussions.  AL is especially grateful to Steve Girvin for discussions, and for suggesting this problem to him. This work was supported by the Department of Energy under Grant DE-SC0024324.   

We acknowledge Codex 5.5 for assistance in writing the code used to perform the GRAPE simulations. We acknowledge ChatGPT 5.4 Thinking for teaching us about cyclotomic polynomials and related  number theory, and ChatGPT 5.5/5.6 Pro for helping to check the proofs.
\begin{appendix}

\makeatletter


\let\@sectioncntformat\@seccntformat

\let\@hangfrom@section\@hang@from


\renewcommand\section{\@ifstar{\AndySectionStar}{\@startsection{section}{1}{\z@}%

  {-3.5ex \@plus -1ex \@minus -.2ex}%

  {2.3ex \@plus .2ex}%

  {\normalfont\Large\bfseries}}}

\renewcommand\subsection{\@ifstar{\AndySubsectionStar}{\@startsection{subsection}{2}{\z@}%

  {-3.25ex\@plus -1ex \@minus -.2ex}%

  {1.5ex \@plus .2ex}%

  {\normalfont\large\bfseries}}}

\renewcommand\subsubsection{\@ifstar{\AndySubsubsectionStar}{\@startsection{subsubsection}{3}{\z@}%

  {-3.25ex\@plus -1ex \@minus -.2ex}%

  {1.5ex \@plus .2ex}%

  {\normalfont\normalsize\bfseries}}}

\makeatother

\renewcommand{\thesubsection}{\thesection.\arabic{subsection}}
\renewcommand{\theequation}{\thesection.\arabic{equation}}
\renewcommand{\thesubsubsection}{\thesubsection.\arabic{subsubsection}}
\makeatletter
\renewcommand\section{\@ifstar{\AndySectionStar}{\@startsection {section}{1}{-\parindent}%
  {-3.5ex \@plus -1ex \@minus -.2ex}%
  {2.3ex \@plus .2ex}%
  {\normalfont\Large\bfseries}}}
\renewcommand\subsection{\@ifstar{\AndySubsectionStar}{\@startsection{subsection}{2}{-\parindent}%
  {-3.25ex\@plus -1ex \@minus -.2ex}%
  {1.5ex \@plus .2ex}%
  {\normalfont\large\bfseries}}}
\renewcommand\subsubsection{\@ifstar{\AndySubsubsectionStar}{\@startsection{subsubsection}{3}{-\parindent}%
  {-3.25ex\@plus -1ex \@minus -.2ex}%
  {1.5ex \@plus .2ex}%
  {\normalfont\normalsize\bfseries}}}
\makeatother

\section{Proof of Theorem \ref{thm:protocol}}\label{app:proof}
\subsection{Initialization of the protocol}
\label{app:initialization-schedule}
We begin by specifying the ideal sequence of states used in the construction. This fixes the initial oscillator--qubit state and the amplitude-transfer parameter $\kappa_k$ at each step of the protocol. For the $N$-fold cat-state protocol, define the target coherent-state
locations by
\begin{equation}
a_k
:=
\alpha
\exp\left(
\frac{2\pi\mathrm{i}k}{N}
\right),
\qquad
k=0,\ldots,N-1.
\label{eq:protocol-target-locations}
\end{equation}
The ideal protocol begins from the initialized state
\begin{equation}
\ket{\Psi_0^{\mathrm{ideal}}}
:=
\begin{pmatrix}
\dfrac{1}{\sqrt{N}}\ket{a_0}
\\[0.7em]
\sqrt{\dfrac{N-1}{N}}\ket{0}
\end{pmatrix}.
\label{eq:protocol-initial-state}
\end{equation}
which can be
prepared exactly from the oscillator vacuum and the spin-up state by
the gate
\begin{equation}
U_{\mathrm{init}}
:=
D\left(\frac{a_0}{2}\right)
D\left(\frac{a_0}{2}Z\right)
\mathrm{e}^{-\mathrm{i}\vartheta_{\mathrm{init}}Y}
\label{eq:protocol-initialization-gate}
\end{equation}
where \begin{equation}
    \vartheta_{\mathrm{init}}
:=
\arctan\sqrt{N-1}.
\end{equation}

Suppose that the first $k$ coherent-state components have already
been prepared, where $k=1,\ldots,N-1$. Immediately before adding the
component at $a_k$, the ideal state has the form
\begin{equation}
\ket{\Psi_{k-1}^{\mathrm{ideal}}}
=
\begin{pmatrix}
\dfrac{1}{\sqrt{N}}
\displaystyle\sum_{j=0}^{k-1}
\ket{a_j}
\\[0.9em]
\sqrt{\dfrac{N-k}{N}}\ket{0}
\end{pmatrix}.
\label{eq:protocol-state-before-step-k}
\end{equation}
Notice that this ideal state is \emph{not} normalized identically t the initial state.  However, this is acceptable to us, because there will be errors in the state preparation; the ``error kets" that track the imperfections in preparing the ideal state will include terms that correct for the normalization of the overall wave function.  Next, choose (recall the definition of $\kappa$ in Definition \ref{def:idealV})
\begin{equation}
\kappa_k
:=
\frac{1}{\sqrt{N-k}},
\qquad
k=1,\ldots,N-1.
\label{eq:protocol-kappa-schedule}
\end{equation}
The ideal action is
\begin{equation}
V_k^{\mathrm{ideal}}
\ket{\Psi_{k-1}^{\mathrm{ideal}}}
=
\begin{pmatrix}
\dfrac{1}{\sqrt{N}}
\displaystyle\sum_{j=0}^{k}
\ket{a_j}
\\[0.9em]
\sqrt{\dfrac{N-k-1}{N}}\ket{0}
\end{pmatrix},
\label{eq:protocol-state-after-step-k}
\end{equation}
where $V_k^{\mathrm{ideal}}$ denotes
$\mathcal{V}_{\mathrm{ideal}}$ with
$a_*=a_k$ and $\kappa=\kappa_k$.

At the final step,
\begin{equation}
\kappa_{N-1}
=
1,
\label{eq:protocol-final-kappa}
\end{equation}
so the lower component is completely emptied and the ideal output is
\begin{equation}
\ket{\Psi_{N-1}^{\mathrm{ideal}}}
=
\begin{pmatrix}
\dfrac{1}{\sqrt{N}}
\displaystyle\sum_{j=0}^{N-1}
\ket{a_j}
\\[0.9em]
0
\end{pmatrix}.
\label{eq:protocol-ideal-final-state}
\end{equation}
Thus, the protocol adds one coherent-state component of amplitude
$1/\sqrt{N}$ at every step. Let $\ket{\Psi_k}$
denote the actual normalized physical state after step $k$. We write
\begin{equation}
\ket{\Psi_k}
=
\ket{\Psi_k^{\mathrm{ideal}}}
+
\ket{\mathcal{E}_k},
\qquad
\langle\Psi_k|\Psi_k\rangle
=
1,
\label{eq:ideal-plus-error-state}
\end{equation}
where $\ket{\mathcal{E}_k}$ contains both the implementation error and
the small correction required to make the physical state normalized.

\subsection{Approximate action of the $V$ gate}
The first step of the proof is to determine the approximate action of the $V$ gate on an input state which takes a suitable form \eqref{eq:Psiin}.  
A short calculation gives: 
\begin{align}
D\!\left(\lambda \mathrm{e}^{\mathrm{i}\phi}Y\right)&D\!\left(-\frac{a_{*}}{2}Z\right)\ket{\Psi_{\mathrm{in}}} \notag \\
&=
\begin{pmatrix}
\frac{D(\lambda \mathrm{e}^{\mathrm{i}\phi})+D(-\lambda \mathrm{e}^{\mathrm{i}\phi})}{2}
& -i\frac{D(\lambda \mathrm{e}^{\mathrm{i}\phi})-D(-\lambda \mathrm{e}^{\mathrm{i}\phi})}{2} \\
i\frac{D(\lambda \mathrm{e}^{\mathrm{i}\phi})-D(-\lambda \mathrm{e}^{\mathrm{i}\phi})}{2}
& \frac{D(\lambda \mathrm{e}^{\mathrm{i}\phi})+D(-\lambda \mathrm{e}^{\mathrm{i}\phi})}{2}
\end{pmatrix}
\begin{pmatrix}
\displaystyle
\sum_j c_jD\!\left(-\frac{a_*}{2}\right)\ket{a_j}
\\[0.7em]
\displaystyle
c_*\left|\frac{a_*}{2}\right\rangle
\end{pmatrix}.
\end{align}
To proceed we invoke the small-$\lambda$ approximation.

\begin{lem}
\label{lem:trig-form} Let $b := \lambda \mathrm{e}^{\mathrm{i}\phi}$, and define
\begin{subequations}\begin{align}
\theta_{j} &:= \mathrm{Im}\bigl(b \overline{a}_{j}\bigr), \\ 
\theta_{*} &:= \mathrm{Im}\bigl(b \overline{a}_{*}\bigr).
\end{align}\end{subequations}
Let
\begin{equation}
\ket{\Psi_{\mathrm{exact}}^{(Y)}}
:=
D\!\left(\lambda \mathrm{e}^{\mathrm{i}\phi}Y\right)D\!\left(-\frac{a_{*}}{2}Z\right)\ket{\Psi_{\mathrm{in}}},
\end{equation}
with $\ket{\Psi_{\mathrm{in}}}$ as above, and define the \emph{trigonometric approximation}
\begin{align}
\ket{\Psi_{\mathrm{trig}}^{(Y)}}
&:=
\sum_{j}c_{j}
\begin{pmatrix}
\cos\bigl(\theta_{*} - 2\theta_{j}\bigr)\\[0.2em]
\sin\bigl(\theta_{*} - 2\theta_{j}\bigr)
\end{pmatrix}
D\!\left(-\frac{a_{*}}{2}\right)\lvert a_{j}\rangle
\;+\;
c_{*}
\begin{pmatrix}
\sin(\theta_{*})\\
\cos(\theta_{*})
\end{pmatrix}
\bigl\lvert\tfrac{a_{*}}{2}\bigr\rangle.
\label{eq:post-Y-trig}
\end{align}
Then, for all $b\in\mathbb{C}$,
\begin{align}
\bigl\||\Psi_{\mathrm{exact}}^{(Y)}\rangle-|\Psi_{\mathrm{trig}}^{(Y)}\rangle\bigr\|
&\le |b|\left(\sum_j|c_j|+|c_*|\right).
\label{eq:lambda-error-explicit}
\end{align}
\end{lem}

\begin{proof}

Set $b=\lambda \mathrm{e}^{\mathrm{i}\phi}$. We have
\begin{equation}
D(bY)
=
\begin{pmatrix}
\frac{D(b)+D(-b)}{2}
&
-\mathrm{i}\frac{D(b)-D(-b)}{2}
\\[0.4em]
\mathrm{i}\frac{D(b)-D(-b)}{2}
&
\frac{D(b)+D(-b)}{2}
\end{pmatrix}.
\end{equation}

We first calculate the upper component directly:
\begin{align}
&
\begin{pmatrix}
\frac{D(b)+D(-b)}{2}
&
-\mathrm{i}\frac{D(b)-D(-b)}{2}
\end{pmatrix}
\begin{pmatrix}
\displaystyle
\sum_j c_jD\!\left(-\frac{a_*}{2}\right)\ket{a_j}
\\[0.7em]
\displaystyle
c_*\left|\frac{a_*}{2}\right\rangle
\end{pmatrix}
\nonumber\\
&\quad=
\sum_j c_j
\frac{D(b)+D(-b)}{2}
D\!\left(-\frac{a_*}{2}\right)\ket{a_j}
-
\mathrm{i}c_*
\frac{D(b)-D(-b)}{2}
\left|\frac{a_*}{2}\right\rangle.
\label{eq:first-row-multiplication}
\end{align}
Using \eqref{eq:coherent-continuity}, we could get for the coherent state centered at $a_j-a_*/2$, the corresponding
phase is
\begin{align}
2\operatorname{Im}
\left[
b\left(a_j-\frac{a_*}{2}\right)^*
\right]
&=
2\operatorname{Im}(ba_j^*)
-
\operatorname{Im}(ba_*^*)
=
2\theta_j-\theta_*.
\end{align}
For the coherent state centered at $a_*/2$, the phase is
\begin{equation}
2\operatorname{Im}
\left[
b\left(\frac{a_*}{2}\right)^*
\right]
=
\theta_*.
\end{equation}
Therefore, the expression in \eqref{eq:first-row-multiplication}
is approximated by
\begin{align}
&
\sum_j c_j
\frac{
\mathrm{e}^{\mathrm{i}(2\theta_j-\theta_*)}
+
\mathrm{e}^{-\mathrm{i}(2\theta_j-\theta_*)}
}{2}
D\!\left(-\frac{a_*}{2}\right)\ket{a_j}
-
\mathrm{i}c_*
\frac{
\mathrm{e}^{\mathrm{i}\theta_*}
-
\mathrm{e}^{-\mathrm{i}\theta_*}
}{2}
\left|\frac{a_*}{2}\right\rangle
\nonumber\\
&=
\sum_j c_j
\cos(2\theta_j-\theta_*)
D\!\left(-\frac{a_*}{2}\right)\ket{a_j}
+
c_*\sin(\theta_*)
\left|\frac{a_*}{2}\right\rangle
\nonumber\\
&=
\sum_j c_j
\cos(\theta_*-2\theta_j)
D\!\left(-\frac{a_*}{2}\right)\ket{a_j}
+
c_*\sin(\theta_*)
\left|\frac{a_*}{2}\right\rangle.
\label{eq:first-row-trigonometric}
\end{align}
This is the upper component of
$\ket{\Psi_{\mathrm{trig}}^{(Y)}}$; the calculation is analogous  for the lower component.

It remains to bound the approximation error.  Consider first one of
the coherent-state components in the upper qubit component, and write
its coherent-state vector simply as $\ket{a}$.  Let $\vartheta$ denote
the phase appearing in the approximation \eqref{eq:coherent-continuity}. Since 
\begin{align}
    \left\|
        D(\pm b)\ket{a}
        -
        \mathrm{e}^{\pm \mathrm{i}\vartheta}\ket{a}
    \right\|
    &\leq
    \sqrt{2\left(1-\mathrm{e}^{-|b|^2/2}\right)},
    \label{eq:Db-positive-error}
\end{align}
we have 
\begin{equation}
    \begin{pmatrix}
        \dfrac{D(b)+D(-b)}{2}\ket{a}
        \\[0.8em]
        \mathrm{i}\dfrac{D(b)-D(-b)}{2}\ket{a}
    \end{pmatrix} \approx     \begin{pmatrix}
        \dfrac{
            \mathrm{e}^{\mathrm{i}\vartheta}
            +
            \mathrm{e}^{-\mathrm{i}\vartheta}
        }{2}\ket{a}
        \\[0.8em]
        \mathrm{i}\dfrac{
            \mathrm{e}^{\mathrm{i}\vartheta}
            -
            \mathrm{e}^{-\mathrm{i}\vartheta}
        }{2}\ket{a}
    \end{pmatrix} + \left(\begin{array}{c} |\mathcal{E}_{\text{trig}}^\uparrow\rangle \\ |\mathcal{E}_{\text{trig}}^\downarrow\rangle\end{array}\right). 
    \label{eq:exact-column-component}
\end{equation}

Applying the triangle inequality over all coherent-state components
gives
\begin{align}
\bigl\|
    \ket{\Psi_{\mathrm{exact}}^{(Y)}}
    -
    \ket{\Psi_{\mathrm{trig}}^{(Y)}}
\bigr\|
&\leq
\left(\sum_j 
|c_j|
\sqrt{
    2\left(1-\mathrm{e}^{-|b|^2/2}\right)
}
+
|c_*|
\sqrt{
    2\left(1-\mathrm{e}^{-|b|^2/2}\right)
}\right)
\nonumber\\
&\leq
|b|
\left(
    \sum_j|c_j|+|c_*|
\right),
\label{eq:total-error-before-b-bound}
\end{align}
which completes the proof.
\end{proof}

\subsection{Phase alignment and the ideal map}

We now impose a phase alignment condition that will allow us to factor out a common direction in the qubit space.

\begin{defn}[Phase alignment condition]
\label{def:phase-alignment}
Let $\kappa\in(0,1]$ and $\delta>0$.  Fix the unique angle
$\theta_\kappa\in(0,\pi/2]$ such that
\begin{equation}
    \mathrm{e}^{\mathrm{i}\theta_\kappa}
    =
    \sqrt{1-\kappa^2}
    +
    \mathrm{i}\kappa.
    \label{eq:2pi-kappa-fix}
\end{equation}
We say that $b=\lambda\mathrm{e}^{\mathrm{i}\phi}$ satisfies the
\emph{$(\kappa,\delta)$--phase-alignment condition} for the data
$\{a_j\},a_*$ if there exist angles $\{\vartheta_j\}$ such that,
for every $j$,
\begin{subequations}
\begin{align}
    2\theta_*-2\vartheta_j
    &=
    \theta_\kappa
    \qquad
    (\mathrm{mod}\ 2\pi),
    \label{eq:phase-align-ideal}
    \\
    |\theta_j-\vartheta_j|
    &\leq
    \delta,
    \label{eq:phase-align-error}
\end{align}
\end{subequations}

where $\theta_{j} = \mathrm{Im}(b a_{j}^{*})$ and $\theta_{*} = \mathrm{Im}(b a_{*}^{*})$.
In other words, the actual phases $\theta_{j}$ are within $\delta$ of some ideal phases $\vartheta_{j}$ that solve the exact alignment condition \eqref{eq:phase-align-ideal}.
\end{defn}

\noindent For a fixed $j$, define the unit vectors in $\mathbb{R}^{2}$:
\begin{subequations}
    \begin{align}
v_{1}(\vartheta_{j}) &:=
\begin{pmatrix}
\cos(\theta_{*}-2\vartheta_{j}) \\
\sin(\theta_{*}-2\vartheta_{j})
\end{pmatrix}, \\
v_{1}^{\perp}(\vartheta_{j}) &:=
\begin{pmatrix}
-\sin(\theta_{*}-2\vartheta_{j}) \\
\cos(\theta_{*}-2\vartheta_{j})
\end{pmatrix}.
\end{align}
\end{subequations}
Using the decomposition along with \eqref{eq:phase-align-ideal}, for each $j$ and each choice of $\vartheta_{j}$, one has the identity 
\begin{equation}
\label{eq:decomposition}
\begin{pmatrix}
\sin(\theta_{*})\\
\cos(\theta_{*})
\end{pmatrix}
=
\sin(2\theta_{*}-2\vartheta_{j})\,v_{1}(\vartheta_{j})
+
\cos(2\theta_{*}-2\vartheta_{j})\,v_{1}^{\perp}(\vartheta_{j}).
\end{equation}

\begin{lem}
\label{lem:ancilla-Lipschitz}
Define
\begin{equation}
u(x) :=
\begin{pmatrix}
\cos(\theta_{*}-2x) \\
\sin(\theta_{*}-2x)
\end{pmatrix}.
\end{equation}
Then for any $x,y\in\mathbb{R}$,
\begin{equation}
\label{eq:ancilla-Lipschitz}
\|u(x) - u(y)\| \le 2\,|x-y|.
\end{equation}
\end{lem}

\begin{proof}
Notice that $\lVert u^\prime(x)\rVert = 2$.  \eqref{eq:ancilla-Lipschitz} then follows from the triangle inequality applied to $u(y) - u(x) = \int\limits_x^y \mathrm{d}z \; u^\prime(z).$
\end{proof}

We now impose the $(\kappa,\delta)$--phase alignment condition. By \eqref{eq:phase-align-ideal}, for all $j$,
\begin{equation}
\sin(2\theta_{*}-2\vartheta_{j})=\kappa,
\end{equation}
and hence 
\begin{equation}
\begin{pmatrix}
\sin(\theta_{*})\\
\cos(\theta_{*})
\end{pmatrix}
=
\kappa\,v_{1}(\vartheta_{j})+
\sqrt{1-\kappa^{2}}\;v_{1}^{\perp}(\vartheta_{j}).
\end{equation}
Next, write
\begin{equation}
e_j:=v_1(\theta_j)-v_1(\vartheta_j). 
\label{eq:v1-split}
\end{equation}
Then by Lemma~\ref{lem:ancilla-Lipschitz} and \eqref{eq:phase-align-error},
\begin{equation}
\label{eq:e-j-bound}
\|e_j\|\le 2|\theta_j-\vartheta_j|\le 2\delta.
\end{equation}

Fix any index $j_*$ corresponding to one of the already-prepared
coherent-state components. Since the phase-alignment condition holds
for every $j$, it holds in particular for $j=j_*$. Setting
$j=j_*$ in Eq.~\eqref{eq:decomposition} gives
\begin{align}
\begin{pmatrix}
\sin\theta_*\\
\cos\theta_*
\end{pmatrix}
&=
\kappa\,v_1(\vartheta_{j_*})
+
\sqrt{1-\kappa^2}\,
v_1^\perp(\vartheta_{j_*}).
\label{eq:special-component-decomposition}
\end{align}
Substituting \eqref{eq:v1-split} and \eqref{eq:special-component-decomposition} into \eqref{eq:post-Y-trig},
we obtain an exact decomposition of the state $|\Psi^{(Y)}_{\text{trig}}\rangle$ into an ideal term plus an error term:
\begin{align}
D\!\left(\lambda \mathrm{e}^{\mathrm{i}\phi}Y\right)
D\!\left(-\frac{a_*}{2}Z\right)\ket{\Psi_{\mathrm{in}}}
&=
\sum_j c_j\,v_1(\theta_j)D\!\left(-\frac{a_{*}}{2}\right)\ket{a_j}
+
c_*
\begin{pmatrix}
\sin\theta_*\\
\cos\theta_*
\end{pmatrix}
\ket{a_*/2} +
\ket{\mathcal E_{\mathrm{trig}}}
\notag \\
&=
\sum_j c_j\,v_1(\vartheta_j)D\!\left(-\frac{a_{*}}{2}\right)\ket{a_j} 
+
\kappa c_*\,v_1(\vartheta_{j_*})\ket{a_*/2} \notag \\
&\;\;\;\;\; +
\sqrt{1-\kappa^2}\,c_*\,v_1^\perp(\vartheta_{j_*})\ket{a_*/2}+
\ket{\mathcal E_{\mathrm{trig}}}
+
\ket{\mathcal E_{\mathrm{phase}}},
\label{eq:postY-ideal-plus-error}
\end{align}
with
\begin{equation}
\ket{\mathcal E_{\mathrm{phase}}}:=\sum_j c_j\,e_jD\left(-\frac{a_*}{2}\right)\ket{a_j},
\end{equation} and \begin{equation}
\label{eq:global-error}
    \|\ket{\mathcal E_{\mathrm{phase}}}\|
\le \sum_j |c_j|\,\|e_j\|
\le 2\delta\sum_j |c_j|.
\end{equation}
We define
\begin{equation}
\ket{\mathcal E}=\ket{\mathcal E_{\mathrm{trig}}}
+
\ket{\mathcal E_{\mathrm{phase}}}
\end{equation}
More concretely, write 
\begin{equation}
    \beta_{j_*}
    :=
    2\vartheta_{j_*}-\theta_*.
\label{eq:rotation_angle}
\end{equation}
Consider the rotation in the ancilla Bloch sphere about the $Y$-axis by angle $-\beta_{j_*}$, which in the $\{|0\rangle,|1\rangle\}$ basis is represented by the real $2\times 2$ matrix
\begin{equation}
    \mathrm{e}^{-\mathrm{i}\beta_{j_*}Y}
    =
    \begin{pmatrix}
        \cos\beta_{j_*} & -\sin\beta_{j_*}\\
        \sin\beta_{j_*} & \cos\beta_{j_*}
    \end{pmatrix}.
\end{equation}
A direct computation shows that
\begin{subequations}
\begin{align}
    \mathrm{e}^{-\mathrm{i}\beta_{j_*}Y}
    v_1(\vartheta_{j_*})
    &=
    \begin{pmatrix}
        1\\
        0
    \end{pmatrix},
    \\
    \mathrm{e}^{-\mathrm{i}\beta_{j_*}Y}
    v_1^\perp(\vartheta_{j_*})
    &=
    \begin{pmatrix}
        0\\
        1
    \end{pmatrix}.
\end{align}
\end{subequations}
Applying the ancilla rotation
$\mathrm{e}^{-\mathrm{i}\beta_{j_*}Y}$ to \eqref{eq:postY-ideal-plus-error}, if the diophantine condition holds we obtain 
\begin{align}
&
\left(
    \mathrm{e}^{-\mathrm{i}\beta_{j_*}Y}
    \otimes I
\right)
D\!\left(\lambda\mathrm{e}^{\mathrm{i}\phi}Y\right)
D\!\left(-\frac{a_*}{2}Z\right)
\ket{\Psi_{\mathrm{in}}}
\nonumber\\
&=
\begin{pmatrix}
\displaystyle
\sum_j c_j D(-a_*/2)\lvert a_j\rangle
+\kappa c_*\lvert a_*/2\rangle
\\[6pt]
\displaystyle
\sqrt{1-\kappa^2}\,c_*\lvert a_*/2\rangle
\end{pmatrix}
+
\ket{\widetilde{\mathcal E}}.
\label{eq:rotated-phase-aligned-state}
\end{align}
where the rotated error ket is
\begin{equation}
    \ket{\widetilde{\mathcal E}}
    :=
    \left(
        \mathrm{e}^{-\mathrm{i}\beta_{j_*}Y}
        \otimes I
    \right)
    \ket{\mathcal E}.
\end{equation}
Since the ancilla rotation is unitary, $\lVert |\widetilde{\mathcal{E}}\rangle \rVert =\lVert |\mathcal{E}\rangle \rVert  $. Finally, we apply the last conditional displacement $D(a_{*}Z/2)$, which simply translates the oscillator coherent states back to $|a_{j}\rangle$ and $|a_{*}\rangle$: 
\begin{align}
V(a_{*},\lambda \mathrm{e}^{\mathrm{i}\phi},\kappa)
\ket{\Psi_{\mathrm{in}}}
&\approx
\begin{pmatrix}
\displaystyle
\sum_{j}c_{j}\lvert a_{j}\rangle
+\kappa c_{*}\lvert a_{*}\rangle
\\[6pt]
\displaystyle
\sqrt{1-\kappa^{2}}\,c_{*}\lvert 0\rangle
\end{pmatrix}.
\label{eq:factored}
\end{align}

\begin{figure}[t]
    \centering
    \includegraphics[width=0.5\textwidth]{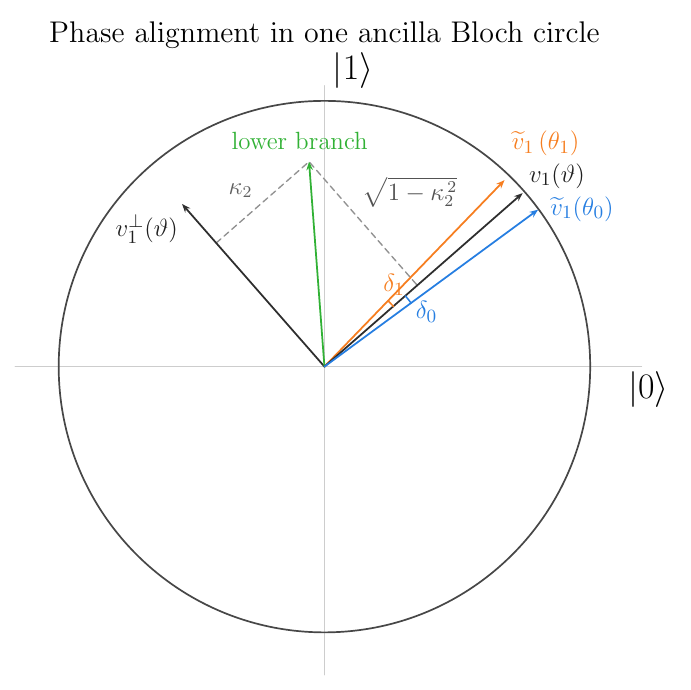}
    \caption{
     View of the phase-alignment condition for one addition gate. The black arrow is the ideal common direction $v_1(\vartheta)$ fixed by the target angle $\theta_{K_2}$, while the blue and orange arrows are the actual ancilla directions $\tilde{v}_1\left(\theta_0\right)$ and $\tilde{v}_1\left(\theta_1\right)$ seen by the two already-prepared peaks. When $b_2$ is chosen correctly, these actual directions differ from the ideal one only by small errors $\delta_0, \delta_1$. The green arrow is the lower-branch ancilla vector, decomposed as $\kappa_2 v_1(\theta)+\sqrt{1-\kappa_2^2} v_1^{\perp}(\vartheta)$. After the subsequent qubit rotation, the common direction $v_1(\vartheta)$ is mapped to the computational basis.
    }
    \label{fig:phase}
\end{figure}

\subsection{Diophantine condition for guaranteeing phase alignment}
We now isolate the Diophantine input needed to guarantee the
approximate phase-alignment condition of
Definition~\ref{def:phase-alignment}.  For angles
$\theta,\theta'\in\mathbb{R}$, define their distance on the unit
circle by
\begin{equation}
\mathsf{d}(\theta,\theta')
:=
\inf_{m\in\mathbb{Z}}
\left|
\theta-\theta'-2\pi m
\right|.
\label{eq:circle-distance}
\end{equation}
\begin{prop}[Diophantine solvability from $\mathbb{Q}$-linear independence]
\label{prop:diophantine-solvability}
Suppose that the relevant phase differences depend linearly on the
real control parameter $\lambda$, so that
\begin{equation}
    2\theta_*-2\theta_j
    =
    \lambda\gamma_j,
    \qquad
    j=1,\ldots,n,
    \label{eq:linear-phase-coefficients}
\end{equation}
for some fixed real numbers $\gamma_1,\ldots,\gamma_n$. If
$\gamma_1,\ldots,\gamma_n$ are linearly independent over
$\mathbb{Q}$, then for every $\varepsilon_{\mathrm{ph}}>0$ there exists
$\lambda\in\mathbb{R}$ such that 
\begin{equation}
    \mathsf{d}\!\left(
        2\theta_*-2\theta_j,\,
        \theta_\kappa
    \right)
    <
    2\pi\varepsilon_{\mathrm{ph}},
    \qquad
    j=1,\ldots,n.
    \label{eq:epsilon-box}
\end{equation}
Thus, the approximate modular condition is solvable whenever the phase
coefficients $\gamma_1,\ldots,\gamma_n$ are
$\mathbb{Q}$-linearly independent.
\end{prop}

\begin{proof}
By linear independence, the orbit $(\lambda \gamma_1,\ldots, \lambda \gamma_n)$
is dense in $\mathbb{T}^n=(\mathbb{R}/2\pi\mathbb{Z})^n$ ~\cite{einsiedler2011ergodic}.  Therefore it must enter the finite-volume subset obeying \eqref{eq:epsilon-box} for some finite $\lambda$.
\end{proof}

\begin{prop}[Phase-misalignment error from the $\varepsilon_{\mathrm{ph}}$-approximate modular condition]
\label{prop:epsilon2-error}
Assume that the approximate modular condition
\eqref{eq:epsilon-box} holds. Then the
$(\kappa,\delta)$--phase-alignment condition of
Definition~\ref{def:phase-alignment} holds with
\begin{equation}
    \delta=\pi\varepsilon_{\mathrm{ph}}.
\end{equation}
Moreover, the phase-alignment error ket $\ket{\mathcal E_{\mathrm{phase}}}$ introduced
in \eqref{eq:global-error} satisfies
\begin{equation}
    \bigl\|\ket{\mathcal E_{\mathrm{phase}}}\bigr\|
    \leq
    2\pi\varepsilon_{\mathrm{ph}}
    \sum_j|c_j|.
    \label{eq:phase-misalignment-bound}
\end{equation}
\end{prop}

 \begin{proof}
By \eqref{eq:epsilon-box} and the definition of the circle distance
in \eqref{eq:circle-distance}, for each $j$ there exists
$m_j\in\mathbb Z$ such that
\begin{equation}
\left|
2\theta_*-2\theta_j-\theta_\kappa-2\pi m_j
\right|
<
2\pi\varepsilon_{\mathrm{ph}}.
\label{eq:phase-mismatch-representative}
\end{equation}
Choose the comparison angle $\vartheta_j$ in
Definition~\ref{def:phase-alignment} to be
\begin{equation}
\vartheta_j
:=
\theta_*-\frac{\theta_\kappa}{2}-\pi m_j.
\end{equation}
Then
\begin{equation}
2\theta_*-2\vartheta_j
=
\theta_\kappa+2\pi m_j,
\end{equation}
and hence the exact alignment condition
\eqref{eq:phase-align-ideal} is satisfied. Furthermore,
\begin{align}
|\theta_j-\vartheta_j|
&=
\frac{1}{2}
\left|
2\theta_*-2\theta_j-\theta_\kappa-2\pi m_j
\right|
\nonumber\\
&<
\pi\varepsilon_{\mathrm{ph}}.
\end{align}
Thus \eqref{eq:phase-align-error} holds with
$\delta=\pi\varepsilon_{\mathrm{ph}}$, proving the
$(\kappa,\pi\varepsilon_{\mathrm{ph}})$--phase-alignment condition.
Finally, substituting $\delta=\pi\varepsilon_{\mathrm{ph}}$ into
\eqref{eq:global-error} gives
\begin{equation}
\bigl\|\ket{\mathcal E_{\mathrm{phase}}}\bigr\|
\leq
2\pi\varepsilon_{\mathrm{ph}}\sum_j|c_j|,
\end{equation}
which is \eqref{eq:phase-misalignment-bound}.
\end{proof}

\subsection{Specialization to equally-spaced coherent states on a circle}

We now address the existence of parameters $b$ satisfying the phase alignment condition. We specialize to coherent states arranged on a circle in phase space.

\begin{defn}[Circular configuration]
Fix $N\ge 2$ and $\alpha\in\mathbb{R}$ with $\alpha> 0$. For $0\le j\le n\le N-2$, define
\begin{subequations}
    \begin{align}
       a_{j} &= \alpha \mathrm{e}^{2\pi \mathrm{i} j/N},  \\
       a_{*} &= \alpha \mathrm{e}^{2\pi \mathrm{i} (n+1)/N}
    \end{align}
\end{subequations}
Then, for some real angle $\theta$, write $b = \lambda \mathrm{e}^{\mathrm{i}\phi}$ and parametrize \begin{subequations}
\label{eq:circular-phase-parametrization}
\begin{align}
\theta_{j}
&= \alpha\lambda\,\mathrm{Im}\left(\mathrm{e}^{\mathrm{i}\phi} \mathrm{e}^{-2\pi \mathrm{i} j/N}\right)
   = \alpha\lambda \sin\!\left(\phi - \frac{2\pi j}{N}\right), \\
\theta_{*}
&= \alpha\lambda\,\mathrm{Im}\left(\mathrm{e}^{\mathrm{i}\phi} \mathrm{e}^{-2\pi \mathrm{i} (n+1)/N}\right)
   = \alpha\lambda \sin\!\left(\phi - \frac{2\pi (n+1)}{N}\right).
\end{align}\end{subequations}
\end{defn}

\noindent Recall that $\theta_\kappa$ was fixed in
\eqref{eq:2pi-kappa-fix}. The phase alignment condition is equivalent to 
\begin{equation}
\theta_{*} - \theta_{j}
= \frac{\theta_\kappa}{2} + \pi m_{j},\qquad m_{j}\in\mathbb{Z}.
\end{equation}

For the purposes of our construction, however, requiring exact phase
alignment is not appropriate; what is needed is an approximate version. Accordingly, we require  
\begin{equation}
\label{eq:mod-approx}
\mathsf{d}\big(
2\theta_{*} - 2\theta_{j},\,\theta_{\kappa}
\big)
< 2\pi\varepsilon_{\mathrm{ph}}
\quad\text{for all }j=0,\dots,n,
\end{equation}
for a prescribed accuracy $\varepsilon_{\mathrm{ph}}>0$.

\begin{prop}[Rational independence and existence of phases]
\label{prop:rational-independence}
Fix a prime $N$ and $0\le n\le N-3$. For $j=0,\dots,n$ let
\begin{equation}
\beta_{j}(\phi)
:
= \frac{\alpha}{\pi}\Biggl(
\sin\!\Bigl(\phi - \frac{2\pi (n+1)}{N}\Bigr)
-
\sin\!\Bigl(\phi - \frac{2\pi j}{N}\Bigr)
\Biggr),
\end{equation}
where $\theta_{*}$ and $\theta_{j}$ are as defined above for $\lambda=1$.
Then for almost every $\phi\in[0,2\pi)$, the real numbers $\beta_0(\phi),\ldots, \beta_n(\phi)$
are linearly independent over $\mathbb{Q}$. In particular, for such a choice of $\phi$ and for any $\theta_{\kappa}\in\mathbb{R}$ and any $\varepsilon_{\mathrm{ph}}>0$, there exists $\lambda\in\mathbb{R}$ such that the approximate modular condition \eqref{eq:mod-approx} holds.
\end{prop}

\begin{proof}
For each nonzero integer vector$
k=(k_0,\ldots,k_n)\in\mathbb Z^{n+1}\setminus\{0\}$, 
define
\begin{equation}
F_k(\phi)
:=
\sum_{j=0}^{n}k_j\beta_j(\phi).
\end{equation}
Set
\begin{equation}
\zeta:=\mathrm{e}^{-2\pi\mathrm{i}/N}.
\end{equation}
Using
\begin{equation}
\sin\left(
\phi-\frac{2\pi m}{N}
\right)
=
\operatorname{Im}\left(
\mathrm{e}^{\mathrm{i}\phi}\zeta^m
\right),
\end{equation}
we obtain
\begin{equation}
F_k(\phi)
=
\frac{\alpha}{\pi}
\operatorname{Im}\left(
\mathrm{e}^{\mathrm{i}\phi}P_k(\zeta)
\right),
\label{eq:Fk-polynomial-form}
\end{equation}
where
\begin{equation}
P_k(z)
:=
\left(
\sum_{j=0}^{n}k_j
\right)z^{n+1}
-
\sum_{j=0}^{n}k_jz^j
\in\mathbb Z[z].
\label{eq:Pk-polynomial}
\end{equation}

Suppose that $F_k$ vanishes identically. Equation
\eqref{eq:Fk-polynomial-form} then implies
\begin{equation}
\operatorname{Im}\left(
\mathrm{e}^{\mathrm{i}\phi}P_k(\zeta)
\right)
=
0
\end{equation}
for every $\phi$, and hence $P_k(\zeta)=0$. However,
\begin{equation}
\deg P_k
\leq
n+1
\leq
N-2
<
N-1
=
\deg\Phi_N.
\end{equation}
Since $\Phi_N$ is the minimal polynomial of $\zeta$ over
$\mathbb Q$, Corollary~\ref{cor:cyclotomic} implies that $P_k$ must
be the zero polynomial. Comparing the coefficients of
$1,z,\ldots,z^n$ then gives
\begin{equation}
k_0=\cdots=k_n=0,
\end{equation}
contrary to the choice of $k$. Thus, $F_k$ is not identically zero
for any nonzero integer vector $k$.

For each such $k$, the zero set
\begin{equation}
Z_k
:=
\left\{
\phi\in[0,2\pi):
F_k(\phi)=0
\right\}
\end{equation}
has Lebesgue measure zero. Since
$\mathbb Z^{n+1}\setminus\{0\}$ is countable, the exceptional set
\begin{equation}
Z
:=
\bigcup_{k\in\mathbb Z^{n+1}\setminus\{0\}}Z_k
\end{equation}
also has measure zero. Therefore, for every
$\phi\in[0,2\pi)\setminus Z$, no nontrivial integer relation among
$\beta_0(\phi),\ldots,\beta_n(\phi)$ exists. These numbers are thus
linearly independent over $\mathbb Q$ for almost every $\phi$.

Finally, for such a $\phi$, we have
\begin{equation}
\theta_{*}(\lambda,\phi)-\theta_{j}(\lambda,\phi)
= \pi\lambda\,\beta_{j}(\phi),
\end{equation}
Therefore, Proposition~\ref{prop:diophantine-solvability}, applied to
the $\mathbb{Q}$-linearly independent numbers
$\beta_0(\phi),\ldots,\beta_n(\phi)$, directly implies that there exists
$\lambda\in\mathbb{R}$ satisfying the approximate modular condition
\eqref{eq:mod-approx}.
\end{proof}

\subsection{The last step of the protocol}
In the last step of the protocol, assuming the errors have not accumulated (which we'll return to at the end),  we have a state of the form
\begin{equation}
\label{eq:final-step-input}
\ket{\Psi_{\mathrm{in}}^{(N-1)}} =\frac{1}{\sqrt{N}}
\begin{pmatrix}
\displaystyle \sum_{m=0}^{N-2}\lvert \alpha\,\mathrm{e}^{2\pi \mathrm{i} m/N}\rangle \\
\lvert 0\rangle
\end{pmatrix},
\end{equation}
and we want to prove we could use $V(\alpha \mathrm{e}^{2\pi \mathrm{i} (N-1)/N},\lambda \mathrm{e}^{\mathrm{i}\phi},1)$ to add the last coherent state $\lvert \alpha\,\mathrm{e}^{2\pi \mathrm{i} (N-1)/N}\rangle$ in the upper component, producing an $N$-fold cat state.  The strategy is the same as before, but we need to be a little bit careful because the locations of the coherent states are no longer linearly independent over $\mathbb{Z}$.  Explicitly, in the last step,  
\begin{align}
&D(\lambda \mathrm{e}^{\mathrm{i}\phi}Y)D\!\left(-\frac{\alpha \mathrm{e}^{2\pi \mathrm{i} (N-1)/N}}{2}Z\right)
\ket{\Psi_{\mathrm{in}}^{(N-1)}} =\frac{1}{\sqrt{N}}D(\lambda \mathrm{e}^{\mathrm{i}\phi}Y)
\begin{pmatrix}
\displaystyle \sum_{m=0}^{N-2} D\!\left(-\frac{\alpha \mathrm{e}^{2\pi \mathrm{i} (N-1)/N}}{2}\right)\lvert \alpha \mathrm{e}^{2\pi \mathrm{i} m/N}\rangle \\
\displaystyle \left\lvert \frac{\alpha \mathrm{e}^{2\pi \mathrm{i} (N-1)/N}}{2}\right\rangle
\end{pmatrix} \nonumber\\[0.4em]
&\quad\approx\frac{1}{\sqrt{N}}
\sum_{m=0}^{N-2}
\begin{pmatrix}
\cos(\theta_{N-1} - 2\theta_{m}) \\
\sin(\theta_{N-1} - 2\theta_{m})
\end{pmatrix}
D\!\left(-\frac{\alpha \mathrm{e}^{2\pi \mathrm{i} (N-1)/N}}{2}\right)\lvert \alpha \mathrm{e}^{2\pi \mathrm{i} m/N}\rangle \notag \\
&\;\;\;\;\;\;\;
+\frac{1}{\sqrt{N}}
\begin{pmatrix}
\sin(\theta_{N-1}) \\
\cos(\theta_{N-1})
\end{pmatrix}
\left\lvert \frac{\alpha \mathrm{e}^{2\pi \mathrm{i} (N-1)/N}}{2}\right\rangle.
\label{eq:last-Y-trig}
\end{align}
To obtain a decoupled form analogous to the ideal map, we seek a choice of phases such that the ancilla vector corresponding to the new coherent state is aligned with that of the existing superposition.

Analogously to the preceding steps, we require that, for every
$m=0,1,\ldots,N-2$,
\begin{equation}
\label{eq:approx-decoupling-condition}
\mathsf{d}\!\left(
2\theta_{N-1}-2\theta_m,\,
\frac{\pi}{2}
\right)
<
2\pi\varepsilon_{\mathrm{ph}},
\end{equation}
where $\mathsf{d}$ is the distance on the unit circle defined in
\eqref{eq:circle-distance}. Thus, for each $m$, there exists
$\delta_m\in\mathbb{R}$ such that
\begin{equation}
2\theta_{N-1}-2\theta_m
=
\frac{\pi}{2}+\delta_m
\qquad
(\mathrm{mod}\ 2\pi),
\end{equation}
with $|\delta_m|<2\pi\varepsilon_{\mathrm{ph}}$. 
It follows that
\begin{equation}
\theta_{N-1}-2\theta_m
=
\frac{\pi}{2}-\theta_{N-1}+\delta_m
\qquad
(\mathrm{mod}\ 2\pi).
\end{equation}

Consequently,
\begin{align}
\left\|
\begin{pmatrix}
\cos(\theta_{N-1}-2\theta_m)\\
\sin(\theta_{N-1}-2\theta_m)
\end{pmatrix}
-
\begin{pmatrix}
\sin(\theta_{N-1})\\
\cos(\theta_{N-1})
\end{pmatrix}
\right\|
&=
\left\|
\begin{pmatrix}
\cos\!\left(\frac{\pi}{2}-\theta_{N-1}+\delta_m\right)\\
\sin\!\left(\frac{\pi}{2}-\theta_{N-1}+\delta_m\right)
\end{pmatrix}
-
\begin{pmatrix}
\cos\!\left(\frac{\pi}{2}-\theta_{N-1}\right)\\
\sin\!\left(\frac{\pi}{2}-\theta_{N-1}\right)
\end{pmatrix}
\right\|
\nonumber\\
\leq{}&
|\delta_m|<
2\pi\varepsilon_{\mathrm{ph}}.
\label{eq:last-step-ancilla-alignment-error}
\end{align}
Using the circular phase parametrization
\eqref{eq:circular-phase-parametrization},
condition \eqref{eq:approx-decoupling-condition} becomes
\begin{equation}
\label{eq:decoupling-Diophantine}
\mathsf{d}\!\left(
2\alpha\lambda
\left[
\sin\!\left(
\phi-\frac{2\pi(N-1)}{N}
\right)
-
\sin\!\left(
\phi-\frac{2\pi m}{N}
\right)
\right],
\frac{\pi}{2}
\right)
<
2\pi\varepsilon_{\mathrm{ph}}
\end{equation}
for every $m=0,1,\ldots,N-2$.

\begin{prop}[Rational independence of the sine phases]
\label{prop:sine-phase-rational-independence}
Let $N$ be prime and define
\begin{equation}
s_m(\phi)
:=
\sin\left(
\phi-\frac{2\pi m}{N}
\right),
\qquad
m=0,\ldots,N-2.
\end{equation}
Then, for almost every $\phi\in[0,2\pi)$, the numbers
\begin{equation}
1,s_0(\phi),\ldots,s_{N-2}(\phi)
\end{equation}
are linearly independent over $\mathbb Q$.
\end{prop}

\begin{proof}
The proof is the same cyclotomic argument used in
Proposition~\ref{prop:rational-independence}. Indeed, for an integer
vector $(\ell,k_0,\ldots,k_{N-2})$, consider
\begin{equation}
F_{\ell,k}(\phi)
:=
\ell+\sum_{m=0}^{N-2}k_ms_m(\phi)
=
\ell+
\operatorname{Im}\left[
\mathrm{e}^{\mathrm{i}\phi}P_k(\zeta)
\right],
\end{equation}
where
\begin{equation}
\zeta:=\mathrm{e}^{-2\pi\mathrm{i}/N},
\qquad
P_k(z):=\sum_{m=0}^{N-2}k_mz^m.
\end{equation}
If $F_{\ell,k}$ vanished identically, then $\ell=0$ and
$P_k(\zeta)=0$. Since
\begin{equation}
\deg P_k\leq N-2<\deg\Phi_N=N-1,
\end{equation}
the minimal-polynomial property implies $P_k=0$, and hence
$k_0=\cdots=k_{N-2}=0$. The same countable-union argument as in
Proposition~\ref{prop:rational-independence} then proves the claim
for almost every $\phi$.
\end{proof}

\begin{prop}[Phase-difference functions]
Fix $N\ge 2$. For $m=0,1,\dots,N-2$ define
\begin{equation}
\label{eq:fm-def}
f_{m}(\phi)
:=
\sin\!\left(\phi - \frac{2\pi (N-1)}{N}\right)
-
\sin\!\left(\phi - \frac{2\pi m}{N}\right).
\end{equation}

If $N$ is prime, then for almost every $\phi$, the numbers
$f_0(\phi),\ldots,f_{N-2}(\phi)$ are linearly independent over
$\mathbb{Q}$. Then \eqref{eq:decoupling-Diophantine} is equivalent to
\begin{equation}
\label{eq:approx-mod-constraint}
\mathsf{d}\!\left(
2\alpha\lambda f_m(\phi),
\frac{\pi}{2}
\right)
<
2\pi\varepsilon_{\mathrm{ph}},
\qquad
m=0,1,\ldots,N-2.
\end{equation}
\end{prop}
\begin{proof} 
Let
\begin{equation}
s(\phi):=
\begin{pmatrix}
s_0(\phi)\\
\vdots\\
s_{N-2}(\phi)
\end{pmatrix},
\qquad
f(\phi):=
\begin{pmatrix}
f_0(\phi)\\
\vdots\\
f_{N-2}(\phi)
\end{pmatrix}.
\end{equation}
Since
\begin{equation}
\sum_{m=0}^{N-1}\sin\!\left(\phi-\frac{2\pi m}{N}\right)=0,
\end{equation}
Then \eqref{eq:fm-def} can be written in matrix form as
\begin{equation}
\label{eq:f-As}
f(\phi)=-(I+J)\,s(\phi),
\end{equation}
where $I$ is the $(N-1)\times(N-1)$ identity matrix and $J$ is the $(N-1)\times(N-1)$ all-ones matrix.  
The matrix $I+J$ is invertible over $\mathbb{Q}$, and we find that
\begin{equation}
\label{eq:s-in-terms-of-f}
s(\phi)=-(I-\tfrac1N J)\,f(\phi).
\end{equation}
In particular, the vectors $f(\phi)$ and $s(\phi)$ are related by an invertible rational linear transformation. 

The functions $\lbrace  s_0,\ldots, s_{N-2}\rbrace $ are linearly independent over $\mathbb{Q}$ by our earlier cyclotomic argument in Proposition \ref{prop:sine-phase-rational-independence}.  Hence, $\lbrace f_0, \ldots, f_{N-2}\rbrace$ are also linearly independent over $\mathbb{Q}$. Therefore, Proposition~\ref{prop:diophantine-solvability} guarantees a finite $\lambda$ satisfying \eqref{eq:approx-mod-constraint}.
\end{proof}

\subsection{Bounding the total error}
\label{subsec:error-accumulation}

We now combine the two local error mechanisms for a single application of the gate $V(a_{*},\theta,\theta_{\kappa})$:
\begin{itemize}
\item \emph{Implementation / trigonometric reduction error} from \eqref{eq:total-error-before-b-bound}.
\item \emph{Phase-alignment (Diophantine) mismatch on the circle} from \eqref{eq:phase-misalignment-bound}.
\end{itemize}
Since all remaining ingredients of $V(a_{*},\lambda \mathrm{e}^{\mathrm{i}\phi},\theta_{\kappa})$ are unitary, the total \emph{single-step} state error for an input of the form $|\Psi_{\mathrm{in}}\rangle$, given in \eqref{eq:Psiin}, 
is bounded by the sum of \eqref{eq:total-error-before-b-bound} and \eqref{eq:phase-misalignment-bound}:
\begin{align}
\label{eq:step-error-raw}
\bigl\|
V(a_{*},\lambda \mathrm{e}^{\mathrm{i}\phi},\theta_{\kappa})\ket{\Psi_{\mathrm{in}}}
-
\mathcal{V}_{\mathrm{ideal}}\ket{\Psi_{\mathrm{in}}}
\bigr\|
&\le
\bigl\|\ket{\mathcal E_{\mathrm{trig}}}\bigr\|
+
\bigl\|\ket{\mathcal E_{\mathrm{phase}}}\bigr\|
\nonumber \\
&\leq
|b| \,
\Bigl(\sum_{j}|c_{j}| + |c_{*}|\Bigr)+ 2\pi\varepsilon_{\mathrm{ph}}\,
\Bigl(\sum_{j}|c_{j}|\Bigr).
\end{align}

The ideal protocol has coefficients
\begin{equation}
    |c_j|=\frac{1}{\sqrt{N}},
    \qquad
    j=1,\ldots,k,
    \qquad
    |c_*|=\sqrt{\frac{N-k}{N}}.
\end{equation}
Consequently,
\begin{subequations}\begin{align}
    \sum_j|c_j|+|c_*|
    &=
    \frac{k+\sqrt{N-k}}{\sqrt{N}}
    \leq
    \sqrt{N},
    \label{eq:protocol-coefficient-bound-total}
    \\
    \sum_j|c_j|
    &=
    \frac{k}{\sqrt{N}}
    \leq
    \sqrt{N}.
    \label{eq:protocol-coefficient-bound-existing}
\end{align}\end{subequations}
Here we used $1\leq N-k$ and hence
$\sqrt{N-k}\leq N-k$.  Thus, at every step of the protocol, the
constants in \eqref{eq:step-error-raw} may be bounded by
\begin{equation}
\label{eq:step-error}
\bigl\|
V_k\ket{\Psi_{\mathrm{in}}^{(k)}}
-
V_k^{\mathrm{ideal}}\ket{\Psi_{\mathrm{in}}^{(k)}}
\bigr\|
\leq
\sqrt{N}\bigl(|b_k|+2\pi\varepsilon_{\mathrm{ph}}\bigr).
\end{equation}
Define
\begin{equation}
t_k
:=
\frac{\alpha\lambda_k}{\pi}.
\label{eq:dimensionless-tk-error}
\end{equation}
For each step, choose an angle $\phi_k$ for which the relevant
phase-difference frequencies are linearly independent over
$\mathbb{Q}$.  By Proposition~\ref{prop:diophantine-solvability},
for every fixed $\varepsilon_{\mathrm{ph}}>0$, there exists a finite value of $t_k$
such that all modular conditions at step $k$ hold with tolerance
$\varepsilon_{\mathrm{ph}}$.
Once $t_k$ has been fixed, the physical displacement is
\begin{equation}
b_k
=
\lambda_k\mathrm{e}^{\mathrm{i}\phi_k}
=
\frac{\pi t_k}{\alpha}
\mathrm{e}^{\mathrm{i}\phi_k},
\label{eq:bk-from-tk-error}
\end{equation}
and therefore
\begin{equation}
|b_k|
=
\frac{\pi|t_k|}{\alpha}.
\label{eq:bk-size-error}
\end{equation}
Let
\begin{equation}
U_{\mathrm{phys}}
:=
V_{N-1}\cdots V_1
\label{eq:physical-circuit-error}
\end{equation}
and
\begin{equation}
U_{\mathrm{ideal}}
:=
V_{N-1}^{\mathrm{ideal}}
\cdots
V_1^{\mathrm{ideal}}.
\label{eq:ideal-circuit-error}
\end{equation}
For a normalized initial state $\ket{\Psi_{\mathrm{in}}}$, the
difference between the physical and ideal circuits can be written as a telescoping sum:
\begin{equation}
U_{\mathrm{phys}}\ket{\Psi_{\mathrm{in}}}
-
U_{\mathrm{ideal}}\ket{\Psi_{\mathrm{in}}}=
\sum_{k=1}^{N-1}
V_{N-1}\cdots V_{k+1}
\left(
V_k-V_k^{\mathrm{ideal}}
\right)
V_{k-1}^{\mathrm{ideal}}
\cdots
V_1^{\mathrm{ideal}}
\ket{\Psi_{\mathrm{in}}}.
\label{eq:telescoping-error}
\end{equation}
The physical gates multiplying each error term from the left and right are
unitary, while the ideal maps on the right produce the ideal state entering the $k$th step, where $V_{k-1}^{\mathrm{ideal}}\cdots V_1^{\mathrm{ideal}} \ket{\Psi_{\mathrm{in}}}=\ket{\Psi_{k-1}^{\mathrm{ideal}}}$.  Taking norms, using the triangle inequality, and applying
Eq.~\eqref{eq:step-error} at every step gives
\begin{equation}
\left\|
U_{\mathrm{phys}}\ket{\Psi_{\mathrm{in}}}
-
U_{\mathrm{ideal}}\ket{\Psi_{\mathrm{in}}}
\right\|
\leq{}
\sqrt{N}
\sum_{k=1}^{N-1}
\left(
|b_k|+2\pi\varepsilon_{\mathrm{ph}}
\right).
\label{eq:total-error-before-t}
\end{equation}
Using Eq.~\eqref{eq:bk-size-error}, we obtain
\begin{equation}
\left\|
U_{\mathrm{phys}}\ket{\Psi_{\mathrm{in}}}
-
U_{\mathrm{ideal}}\ket{\Psi_{\mathrm{in}}}
\right\|
\le 
\frac{\pi\sqrt{N}}{\alpha}
\sum_{k=1}^{N-1}|t_k|
+
2\pi(N-1)\sqrt{N}\,\varepsilon_{\mathrm{ph}}.
\label{eq:total-error-alpha}
\end{equation}
For fixed $N$ and fixed phase tolerance $\varepsilon_{\mathrm{ph}}$, define
\begin{equation}
T_N(\varepsilon_{\mathrm{ph}})
:=
\sum_{k=1}^{N-1}|t_k|.
\label{eq:TN-error}
\end{equation}
Since the protocol has finitely many steps and every $t_k$ is finite,
we have
\begin{equation}
T_N(\varepsilon_{\mathrm{ph}})
<
\infty.
\label{eq:TN-error-finite}
\end{equation}
Thus,
\begin{equation}
\left\|
U_{\mathrm{phys}}\ket{\Psi_{\mathrm{in}}}
-
U_{\mathrm{ideal}}\ket{\Psi_{\mathrm{in}}}
\right\|
\leq
\frac{\pi\sqrt{N}}{\alpha}
T_N(\varepsilon_{\mathrm{ph}})
+
2\pi(N-1)\sqrt{N}\,\varepsilon_{\mathrm{ph}}.
\label{eq:total-error-TN-rho}
\end{equation}
Let $h>0$ be the prescribed tolerance for the total state error,
and fix an error-budget parameter
\begin{equation}
0<\rho<1.
\label{eq:rho-range-error}
\end{equation}
We allocate the fraction $\rho$ of the total error budget to the
phase-alignment error.  Thus, choose $\varepsilon_{\mathrm{ph}}$ so that
\begin{equation}
2\pi(N-1)\sqrt{N}\,\varepsilon_{\mathrm{ph}}
=
\rho h.
\label{eq:phase-error-rho-budget}
\end{equation}
Equivalently,
\begin{equation}
\varepsilon_{\mathrm{ph}}
=
\frac{\rho h}
{2\pi(N-1)\sqrt{N}}.
\label{eq:epsilon-rho-choice}
\end{equation}
The remaining error budget for the implementation error is
$(1-\rho)h$.  It is therefore sufficient to require
\begin{equation}
\frac{\pi\sqrt{N}}{\alpha}
T_N(\varepsilon_{\mathrm{ph}})
<
(1-\rho)h.
\label{eq:implementation-rho-budget}
\end{equation}
Solving for $\alpha$ gives
\begin{equation}
\alpha
>
\frac{\pi\sqrt{N}}
{(1-\rho)h}
T_N(\varepsilon_{\mathrm{ph}}).
\label{eq:alpha-rho-epsilon}
\end{equation}

\subsection{Time cost of the protocol}
\label{subsec:phase-space-runtime}

We next estimate the control cost of our constructive protocol. We use
the dimensionless pulse-area cost introduced in
Section~\ref{sec:review}. More precisely, for a circuit $U$, define
\begin{equation}
t_{\mathrm{run}}(U)
:=
\sum_{g\in\mathcal{D}(U)}
|\beta_g|
+
\sum_{r\in\mathcal{R}(U)}
|\vartheta_r|,
\label{eq:runtime-cost-definition-section5}
\end{equation}
where $\mathcal{D}(U)$ is the set of unconditional and
qubit-dependent displacement gates in the circuit, $\beta_g$ is the
displacement amplitude of gate $g$, $\mathcal{R}(U)$ is the set of
single-qubit rotations, and $\vartheta_r$ is the rotation angle of
gate $r$. By calling this the runtime, we are assuming that the coefficients of the Hamiltonian \eqref{eq:phasespace-H} are bounded. 

The initialization used in the constructive protocol was given in \eqref{eq:protocol-initialization-gate}.
The pulse-area cost of the initialization is
\begin{equation}
t_{\mathrm{run}}
\left(
U_{\mathrm{init}}
\right)
=
|a_0|
+
|\vartheta_{\mathrm{init}}|
=
|\alpha|
+
|\vartheta_{\mathrm{init}}|.
\label{eq:section5-initialization-cost}
\end{equation}

If the cost of each $V_k$ were evaluated separately, the two
conditional displacements of size $|a_k|/2$ would appear to give a
total contribution of order $N|\alpha|$. This estimate is not the
cost of the compiled circuit, because conditional displacements in
neighboring blocks can be merged: in the product $V_{k+1}V_k$, we have
\begin{equation}
D\left(
-\frac{a_{k+1}}{2}Z
\right)
D\left(
\frac{a_k}{2}Z
\right)
=
D\left(
\frac{a_k-a_{k+1}}{2}Z
\right).
\label{eq:section5-neighboring-displacements-merge}
\end{equation}
The merged displacements connect neighboring vertices of the regular
$N$-gon. Their magnitudes are
\begin{equation}
\left|
\frac{a_k-a_{k+1}}{2}
\right|
=
|\alpha|
\sin\left(
\frac{\pi}{N}
\right).
\label{eq:section5-polygon-edge-length}
\end{equation}
There are $N-2$ such interior merged displacements. Including the two
unmerged endpoint displacements and the two displacements in
$U_{\mathrm{init}}$, the total contribution from displacements whose
amplitudes are proportional to $|\alpha|$ is bounded by
\begin{equation}
t_{\text{run, disp}}
\leq
\frac{|a_{N-1}|}{2}
+
\sum_{k=1}^{N-2}
\frac{|a_k-a_{k+1}|}{2}
+
\frac{|a_1|}{2}
+
|a_0|
\\
=
2|\alpha|
+
(N-2)|\alpha|
\sin\left(
\frac{\pi}{N}
\right).
\label{eq:section5-large-displacement-cost}
\end{equation}
Using $\sin |x| \le |x|$ we obtain
\begin{equation}
t_{\text{run, disp}}
\leq
\left(
2+\pi
\right)
|\alpha|.
\label{eq:section5-large-displacement-cost-bound}
\end{equation}
Thus, the large conditional displacements trace the perimeter of the
polygon rather than repeatedly traveling a distance of order
$|\alpha|$ at every protocol step.

For qubit rotations, each rotation angle may be represented by a
value $\theta_k \in [-\pi, \pi]$.   
The total single-qubit-rotation contribution consequently obeys
\begin{equation}
t_{\text{run, rot}}
\leq
\pi N.
\label{eq:section5-rotation-cost}
\end{equation}

The remaining displacement contribution comes from the small
$Y$-controlled kicks:
\begin{equation}
t_{\text{run, small}}
=
\sum_{k=1}^{N-1}
|b_k|.
\label{eq:section5-small-displacement-cost}
\end{equation}
Combining
\eqref{eq:bk-size-error}, \eqref{eq:section5-large-displacement-cost},
\eqref{eq:section5-rotation-cost}, and
\eqref{eq:section5-small-displacement-cost} gives the explicit bound

\begin{equation}
t_{\mathrm{run}}
\leq
\left(
2+\pi
\right)
|\alpha|
+
\pi N
+
\frac{\pi}{
|\alpha|
}
T_N\left(
\varepsilon_{\mathrm{ph}}
\right).
\label{eq:section5-runtime-TN}
\end{equation}
Thus, for fixed $N$ and fixed phase-alignment tolerance, we obtain \eqref{eq:trun-good}.
\subsection{Extension to generalized cat states}
\label{app:general_initial}
The same construction can be extended to generalized cat states with
nonuniform amplitudes and phases. Let
\begin{equation}
u_k:=c_k\mathrm{e}^{\mathrm{i}\Gamma_k},\qquad k=0,\ldots,N-1,
\label{eq:generalized-coefficient}
\end{equation}
where $c_k>0$ is the magnitude of the $k$th coefficient and
$\Gamma_k\in\mathbb{R}$ is its phase. Because an overall phase is physically irrelevant, we may assume $\Gamma_0=0$. Because of the normalization, we assume
\begin{equation}
\sum_{k=0}^{N-1}|u_k|^2=\sum_{k=0}^{N-1}c_k^2=1.
\label{eq:generalized-coefficient-normalization}
\end{equation}
In the large-separation regime, the generalized target state is
\begin{equation}
\ket{\mathrm{cat}_{\mathrm{gen}}}=\sum_{k=0}^{N-1}u_k\ket{a_k}=\sum_{k=0}^{N-1} c_k\mathrm{e}^{\mathrm{i}\Gamma_k}\ket{a_k}.
\label{eq:generalized-cat-state}
\end{equation}
Define the remaining amplitude reservoir before step $k$ by
\begin{equation}
r_k:=\sqrt{\sum_{\ell=k}^{N-1}c_\ell^2},\qquad k=0,\ldots,N,
\label{eq:generalized-reservoir}
\end{equation}
The generalized protocol begins from
\begin{equation}
\ket{\Psi_0^{\mathrm{gen}}}:=
\begin{pmatrix}u_0\ket{a_0}\\[0.7em]
r_1\ket{0}
\end{pmatrix}.
\label{eq:generalized-initial-state}
\end{equation}
Before adding the component at $a_k$, the generalized ideal state is
\begin{equation}
\ket{\Psi_{k-1}^{\mathrm{gen}}}=
\begin{pmatrix}
\sum_{j=0}^{k-1}u_j\ket{a_j}\\[0.7em]
r_k\ket{0}
\end{pmatrix}.
\label{eq:generalized-state-before-step}
\end{equation}
At step $k$, we could find
\begin{equation}
\kappa_k=\frac{c_k}{r_k}=\frac{c_k}{\sqrt{\displaystyle\sum_{\ell=k}^{N-1}c_\ell^2}},
\label{eq:generalized-kappa-and-phase}
\end{equation}
which has the same meaning of \eqref{eq:2pi-kappa-fix}.Let $V_k$ denote the ideal real-amplitude transfer map with
$a_*=a_k$ and $\kappa=\kappa_k$ as in \eqref{eq:V-operator}
\begin{equation}
V_k
\begin{pmatrix}
\sum_{j=0}^{k-1}u_j\ket{a_j}\\[0.7em]
r_k\ket{0}
\end{pmatrix}
=
\begin{pmatrix}
\sum_{j=0}^{k-1}u_j\ket{a_j}+\kappa_kr_k\ket{a_k}\\[0.7em]
\sqrt{1-\kappa_k^2}r_k\ket{0}
\end{pmatrix}.
\label{eq:generalized-real-transfer}
\end{equation}
This map produces the correct magnitude $c_k=\kappa_kr_k$, but it does
not yet produce the phase $\mathrm{e}^{\mathrm{i}\Gamma_k}$.
To introduce this phase, define the generalized transfer gate by
\begin{equation}
V_k^{\mathrm{gen}}:=\mathrm{e}^{\mathrm{i}\Gamma_k Z/2}V_k\mathrm{e}^{-\mathrm{i}\Gamma_k Z/2}.
\label{eq:generalized-gate-conjugation}
\end{equation}
Although this expression contains two explicit $Z$ rotations, conjugation only rotates the axes of the existing qubit rotation and $Y$-controlled displacement, while leaving the $Z$-controlled displacements unchanged. Thus, $V_k^{\mathrm{gen}}$ still requires only four instruction-set gates, and the $4N$ depth bound is unchanged. Calculating it explicitly, we get
Finally, applying $\mathrm{e}^{\mathrm{i}\Gamma_k Z/2}$ gives
\begin{align}
V_k^{\mathrm{gen}}
\begin{pmatrix}
\sum_{j=0}^{k-1} u_j\ket{a_j}
\\[0.7em]
r_k\ket{0}
\end{pmatrix}
&=\mathrm{e}^{\mathrm{i}\Gamma_k Z/2}V_k\mathrm{e}^{-\mathrm{i}\Gamma_k Z/2}
\begin{pmatrix}
\sum_{j=0}^{k-1} u_j\ket{a_j}\\[0.7em]
r_k\ket{0}
\end{pmatrix}
\nonumber=\mathrm{e}^{\mathrm{i}\Gamma_k Z/2}V_k
\begin{pmatrix}
\mathrm{e}^{-\mathrm{i}\Gamma_k/2}\sum_{j=0}^{k-1} u_j\ket{a_j}\\[0.7em]
\mathrm{e}^{\mathrm{i}\Gamma_k/2}r_k\ket{0}
\end{pmatrix}\\
&=\begin{pmatrix}
\sum_{j=0}^{k-1} u_j\ket{a_j}+\mathrm{e}^{\mathrm{i}\Gamma_k}\kappa_kr_k\ket{a_k}\\[0.9em]
\sqrt{1-\kappa_k^2}r_k\ket{0}
\end{pmatrix}.
\label{eq:generalized-conjugated-action}
\end{align}
At the final step, we have $\kappa_{N-1}=\frac{c_{N-1}}{r_{N-1}}=1$. The lower component is therefore completely emptied, and the ideal output becomes
\begin{equation}
\ket{\Psi_{N-1}^{\mathrm{gen}}}=
\begin{pmatrix}
\displaystyle\sum_{j=0}^{N-1}u_j\ket{a_j}\\[0.9em]
0
\end{pmatrix}=
\begin{pmatrix}
\displaystyle\sum_{j=0}^{N-1}c_j\mathrm{e}^{\mathrm{i}\Gamma_j}\ket{a_j}\\[0.9em]
0
\end{pmatrix}.
\label{eq:generalized-final-state}
\end{equation}

\section{Explicit protocols}
\subsection{Two explicit protocols for approximate 5-fold cat states}\label{app:N5}
We provide the gate sequences for the \(N=5\) protocol.
For convenience, we group the results into two families of gates acting on the initial state $\ket{0}$. In $U^{(i)}$, the index $(i)$ specifies the Diophantine condition used: $U^{(1)}$ corresponds to $\varepsilon_{\mathrm{ph}}=0.029$, whereas $U^{(2)}$ corresponds to $\varepsilon_{\mathrm{ph}}=0.065$. Both of these gate sequences act on the state $|0\rangle\otimes|0\rangle$.
\begin{align}
U_{N=5}^{(1)}(\alpha)={}&
D\!\left(\frac{\alpha \mathrm{e}^{8\pi \mathrm{i}/5}}{2}Z\right)\,
\mathrm{e}^{-\mathrm{i}(2.190)Y}\,
D\!\left(\frac{5.373+1.746\,\mathrm{i}}{\alpha}Y\right)\,
D\!\left(-\frac{\alpha \mathrm{e}^{8\pi \mathrm{i}/5}}{2}Z\right) \notag \\
&\times
D\!\left(\frac{\alpha \mathrm{e}^{6\pi \mathrm{i}/5}}{2}Z\right)\,
\mathrm{e}^{-\mathrm{i}(2.557)Y}\,
D\!\left(\frac{0.411-5.324\,\mathrm{i}}{\alpha}Y\right)\,
D\!\left(-\frac{\alpha \mathrm{e}^{6\pi \mathrm{i}/5}}{2}Z\right) \notag \\
&\times
D\!\left(\frac{\alpha \mathrm{e}^{4\pi \mathrm{i}/5}}{2}Z\right)\,
\mathrm{e}^{-\mathrm{i}(3.135)Y}\,
D\!\left(\frac{-4.793+6.597\,\mathrm{i}}{\alpha}Y\right)\,
D\!\left(-\frac{\alpha \mathrm{e}^{4\pi \mathrm{i}/5}}{2}Z\right) \notag \\
&\times
D\!\left(\frac{\alpha \mathrm{e}^{2\pi \mathrm{i}/5}}{2}Z\right)\,
\mathrm{e}^{-\mathrm{i}(5.545)Y}\,
D\!\left(\frac{-1.002+1.000\,\mathrm{i}}{\alpha}Y\right)\,
D\!\left(-\frac{\alpha \mathrm{e}^{2\pi \mathrm{i}/5}}{2}Z\right) \notag \\
&\times
D\!\left(\frac{\alpha}{2}\right)\,
D\!\left(\frac{\alpha}{2}Z\right)\mathrm{e}^{-\mathrm{i}(\arctan 2)Y}. \\ 
U_{N=5}^{(2)}(\alpha)={}&
D\!\left(\frac{\alpha \mathrm{e}^{8\pi \mathrm{i}/5}}{2}Z\right)\,
\mathrm{e}^{-\mathrm{i}(2.790)Y}\,
D\!\left(\frac{-2.783+4.419\,\mathrm{i}}{\alpha}Y\right)\,
D\!\left(-\frac{\alpha \mathrm{e}^{8\pi \mathrm{i}/5}}{2}Z\right) \notag \\ 
&\times
D\!\left(\frac{\alpha \mathrm{e}^{6\pi \mathrm{i}/5}}{2}Z\right)\,
\mathrm{e}^{-\mathrm{i}(2.252)Y}\,
D\!\left(\frac{-5.097-1.839\,\mathrm{i}}{\alpha}Y\right)\,
D\!\left(-\frac{\alpha \mathrm{e}^{6\pi \mathrm{i}/5}}{2}Z\right) \notag \\
&\times
D\!\left(\frac{\alpha \mathrm{e}^{4\pi \mathrm{i}/5}}{2}Z\right)\,
\mathrm{e}^{-\mathrm{i}(2.267)Y}\,
D\!\left(\frac{-3.141+4.324\,\mathrm{i}}{\alpha}Y\right)\,
D\!\left(-\frac{\alpha \mathrm{e}^{4\pi \mathrm{i}/5}}{2}Z\right) \notag \\
&\times
D\!\left(\frac{\alpha \mathrm{e}^{2\pi \mathrm{i}/5}}{2}Z\right)\,
\mathrm{e}^{-\mathrm{i}(5.545)Y}\,
D\!\left(\frac{-1.002+1.000\,\mathrm{i}}{\alpha}Y\right)\,
D\!\left(-\frac{\alpha \mathrm{e}^{2\pi \mathrm{i}/5}}{2}Z\right) \notag \\
&\times
D\!\left(\frac{\alpha}{2}\right)\,
D\!\left(\frac{\alpha}{2}Z\right)\mathrm{e}^{-\mathrm{i}(\arctan 2)Y}.
\end{align}

\subsection{An explicit phase-alignment solution for $N=3$}
\label{app:N3-exact-solution}
For $N=3$, the phase-alignment conditions can be solved exactly. The protocol begins from
\begin{equation}
\ket{\Psi_0}=
\begin{pmatrix}
\dfrac{1}{\sqrt{3}}\ket{a_0}\\[0.7em]
\sqrt{\dfrac{2}{3}}\ket{0}
\end{pmatrix}.
\end{equation}
where $a_k$ is defined in \eqref{eq:protocol-target-locations}.
At the first step, $\kappa_1=\frac{1}{\sqrt{2}}$ as in \eqref{eq:protocol-kappa-schedule}. We could choose
\begin{equation}
b_1=-\frac{\mathrm{i}\pi}{12\alpha}.
\label{eq:N3-b1-short}
\end{equation}
Then as in \eqref{eq:circular-phase-parametrization}, it gives
\begin{equation}
\theta_0^{(1)}=-\frac{\pi}{12},\qquad\theta_1^{(1)}=\frac{\pi}{24},
\end{equation}
and hence the diophantine condition is satisfied as in \eqref{eq:phase-align-error}. The corresponding qubit-rotation angle \eqref{eq:rotation_angle} is
\begin{equation}
\beta_1=\theta_1^{(1)}-2\theta_0^{(1)}=\frac{5\pi}{24}.
\end{equation}
Thus, with $V_1=D\left(\frac{a_1}{2}Z\right)\mathrm{e}^{\mathrm{i}\beta_1Y}D(b_1Y)D\left(-\frac{a_1}{2}Z\right)$, we obtain
\begin{equation}
V_1\ket{\Psi_0}
\approx
\begin{pmatrix}
\dfrac{1}{\sqrt{3}}
\left(
\ket{a_0}+\ket{a_1}
\right)
\\[0.8em]
\dfrac{1}{\sqrt{3}}\ket{0}
\end{pmatrix}.
\end{equation}
At the final step, we know $\kappa_2=1$. Choosing
\begin{equation}
b_2=\frac{\pi}{6\alpha}\mathrm{e}^{-\mathrm{i}\pi/6}.
\label{eq:N3-b2-short}
\end{equation}
Then it gives
\begin{equation}
\theta_0^{(2)}=\theta_1^{(2)}=-\frac{\pi}{12},\qquad\theta_2^{(2)}=\frac{\pi}{6}.
\end{equation}
The diophantine condition is also satisfied. The required qubit-rotation angle is
\begin{equation}
\beta_2=\theta_2^{(2)}-2\theta_0^{(2)}=\frac{\pi}{3}.
\end{equation}
After $V_2=D\left(\frac{a_2}{2}Z\right)\mathrm{e}^{\mathrm{i}\beta_2Y}D(b_2Y)D\left(-\frac{a_2}{2}Z\right)$, we get
\begin{equation}
V_2V_1\ket{\Psi_0}
\approx
\begin{pmatrix}
\dfrac{1}{\sqrt{3}}
\left(
\ket{a_0}+\ket{a_1}+\ket{a_2}
\right)
\\[0.8em]
0
\end{pmatrix}.
\end{equation}
The phase-misalignment error vanishes at both steps. The remaining
trigonometric-reduction error satisfies
\begin{equation}
\left\|
\ket{\mathcal E_{\mathrm{trig}}}
\right\|
\leq
\frac{\pi(7+\sqrt{2})}
{12\sqrt{3}\,\alpha},
\end{equation}
The explicit protocol for this is 
\begin{align}
U_{N=3}(\alpha)
={}&
D\!\left(
\frac{\alpha\mathrm{e}^{4\pi\mathrm{i}/3}}{2}Z
\right)\,
\mathrm{e}^{\mathrm{i}(\pi/3)Y}\,
D\!\left(
\frac{\pi/(4\sqrt{3})-\mathrm{i}\pi/12}{\alpha}Y
\right)\,
D\!\left(
-\frac{\alpha\mathrm{e}^{4\pi\mathrm{i}/3}}{2}Z
\right)
\notag\\
&\times
D\!\left(
\frac{\alpha\mathrm{e}^{2\pi\mathrm{i}/3}}{2}Z
\right)\,
\mathrm{e}^{\mathrm{i}(5\pi/24)Y}\,
D\!\left(
-\frac{\mathrm{i}\pi}{12\alpha}Y
\right)\,
D\!\left(
-\frac{\alpha\mathrm{e}^{2\pi\mathrm{i}/3}}{2}Z
\right)
\notag\\
&\times
D\!\left(\frac{\alpha}{2}\right)\,
D\!\left(\frac{\alpha}{2}Z\right)\,
\mathrm{e}^{-\mathrm{i}(\arctan\sqrt{2})Y}
\label{eq:N3-explicit-protocol}
\end{align}
\end{appendix}
acting on the state $|0\rangle\otimes|0\rangle$.
\phantomsection
\section*{References}
\addcontentsline{toc}{section}{References}
\makeatletter
\let\bibsection\relax
\makeatother
\bibliography{ref}
\end{document}